\documentclass[11pt]{article}
\usepackage{amsmath}
\usepackage{amsfonts}
\usepackage{amsxtra}
\usepackage{amssymb}
\usepackage{amsthm}
\usepackage{fullpage}
\usepackage{bbm,xspace}
\usepackage{palatino}

\usepackage{graphicx}
\usepackage[singlelinecheck=false,
            justification=justified,
            labelfont={bf},
            labelsep=quad]
            {caption}
\usepackage{floatrow}

\usepackage{mathtools}
\usepackage{stmaryrd}
\usepackage{euscript}
\usepackage{paralist}
\usepackage{microtype}

\usepackage[%
           colorlinks=true,
           pdfpagemode=UseNone,
           urlcolor=blue,
           linkcolor=blue,
           citecolor=blue,
           pdfstartview=FitH]
           {hyperref}           
\usepackage{fr-fancy}
\usepackage{cleveref}
\makeatletter
\global\let\@myarticle\@true
\makeatletter
\renewcommand{\paragraph}{%
  \@startsection{paragraph}{4}%
  {\z@}{0.0ex \@plus 1ex \@minus .2ex}{-0.25em}%
  {\normalfont\normalsize\bfseries}%
}
\makeatother

\def\triangleq{\overset{\mathrm{def}}{=}}

\newcommand{\dec}[1]{#1\!\!\downarrow}

\let\eps=\varepsilon

\newtheorem{mytheorem}{Theorem}[section]

\newtheorem{myclaim}[mytheorem]{Claim}
\newtheorem{mycorollary}[mytheorem]{Corollary}
\newtheorem{definition}[mytheorem]{Definition}

\newtheorem{mylemma}[mytheorem]{Lemma}
\newtheorem{notation}[mytheorem]{Notation}

\newtheorem{myproposition}[mytheorem]{Proposition}

\nonstopmode

\newtheorem{remark}[mytheorem]{Remark}

\crefname{assumption}{Assumption}{Assumptions}
\crefname{conjecture}{Conjecture}{Conjectures}
\crefname{myclaim}{Claim}{Claims}
\crefname{mycorollary}{Corollary}{Corollaries}
\crefname{definition}{Definition}{Definitions}
\crefname{example}{Example}{Examples}
\crefname{exercise}{Exercise}{Exercises}
\crefname{fact}{Fact}{Facts}
\crefname{mylemma}{Lemma}{Lemmas}
\crefname{notation}{Notation}{Notations}
\crefname{note}{Notes}{Notes}
\crefname{observation}{Observation}{Observations}
\crefname{myproposition}{Proposition}{Propositions}
\crefname{problem}{Problem}{Problems}
\crefname{question}{Question}{Questions}
\crefname{remark}{Remark}{Remarks}

\newcommand{\R}{\mathbb{R}}

\newcommand{\symm}{\mathbb{S}}
\newcommand{\sympm}{\mathbb{S}_{+}}

\def\sm{\setminus}
\let\proj=\Pi

\newcommand{\unit}[1]{\overline{#1}}
\newcommand{\dunit}[1]{\overline{\overline{#1}}}
\DeclareMathOperator{\spn}{span}
\newcommand{\ind}[1]{\mathbf{e}_{#1}}

\DeclareMathOperator{\rank}{rank}
\let\es=\emptyset

\newcommand{\bestS}{{\mathcal{S}^\ast}}
\newcommand{\bestSfi}[1]{\|\xvec_{\circ| \bestS(\circ)}\|^2}
\newcommand{\xvec}{\mathrm{x}}%
\def\opt{\mathsf{OPT}}

\DeclareMathOperator*{\st}{st}
\DeclareMathOperator*{\argmin}{argmin}
\DeclareMathOperator*{\argmax}{argmax}

\newcommand{\expct}[2]{\mathbb{E}_{#1}\bigg[ #2 \bigg]}

\newcommand{\expcts}[2]{\mathbb{E}_{#1}\big[ #2 \big]}

 \makeatother

\def\minmaxkpar{\textsc{$k$-expansion}} %

\def\optg{\Gamma_*}

\def\curg{\Gamma}
\def\newg{\widehat{\Gamma}}

\def\moptg{{\mathbf{\optg}}}

\def\mcurg{{\mathbf{\Gamma}}}
\def\mnewg{{\mathbf{\widehat{\Gamma}}}}

\def\simfun{\Delta} %

\def\setv{\mathsf{Set}_V}

\def\disjv{\mathsf{Disj}_V}

\newcommand{\orth}[1]{\mathcal{S}_{#1}(\R^n)}
\newcommand{\orthk}{\orth{k}}

\setdefaultleftmargin{2.5em}{2.2em}{1.87em}{1.7em}{1em}{1em}
\crefname{program}{algorithm}{algorithms}

\DeclareNewFloatType{program}{placement=h,name=Algorithm}
\begin{document}
\title{How to Round Subspaces: A New Spectral Clustering Algorithm}

\author{Ali Kemal Sinop\footnote{ 
    Email: \url{%
      asinop@cs.cmu.edu}.  This work was done while author was
    visiting Simons Institute for the Theory of Computing, University
    of California at Berkeley.  
    This material is based upon work
    supported by the National Science Foundation under Grant
    No. 1540685.  }} \date{\today}

\maketitle

\begin{abstract}
  A basic problem in spectral clustering is the following. If a
  solution obtained from the spectral relaxation is close to an
  integral solution, is it possible to find this integral solution
  even though they might be in completely different basis?
  In this paper, we propose a new spectral clustering algorithm. It
  can recover a $k$-partition such that the subspace corresponding to
  the span of its indicator vectors is $O(\sqrt{\opt})$ close to the
  original subspace in spectral norm with $\opt$ being the minimum
  possible ($\opt \le 1$ always). Moreover our algorithm does not
  impose any restriction on the cluster sizes. Previously, no
  algorithm was known which could find a $k$-partition closer than
  $o(k \cdot \opt)$.

  We present two applications for our algorithm. First one finds a
  disjoint union of bounded degree expanders which approximate a given
  graph in spectral norm. The second one is for approximating the
  sparsest $k$-partition in a graph where each cluster have expansion
  at most $\phi_k$ provided $\phi_k \le O(\lambda_{k+1})$ where
  $\lambda_{k+1}$ is the $(k+1)^{st}$ eigenvalue of Laplacian
  matrix. This significantly improves upon the previous algorithms,
  which required $\phi_k \le O(\lambda_{k+1}/k)$.
\end{abstract}

\section{Introduction}
\label{sec:intro}
In this paper, we study the following problem. If the solution of
spectral relaxation for some $k$-way partitioning problem is close to
an integral solution, can we still find this integral solution? The
main difficulty is due to the rotational invariance of the spectral
relaxation. The basis of an integral solution might be completely
different than the basis of solutions for the spectral
relaxation. Arguably, this is an important problem in spectral
clustering, which is a widely used approach for many data clustering
and graph partitioning problems arising in practice.

In spectral clustering, one uses the top (or bottom) $k$-eigenvectors
of some matrix derived from the input (usually the Laplacian or
adjacency matrix of some graph derived from the distances or nearest
neighbors) to find a $k$-partition.
%
\iffalse
%
Usually, there is a close connection between spectral methods and the
basic SDP (semi-definite programming) relaxations of associated
clustering and partitioning problems: Top (or bottom) $k$ eigenvectors
correspond to an optimal solution of these SDP formulations.
%
\fi 
%
If the clusters are separated in a nice way, then these
$k$-eigenvectors will be close a $k$-partition up to an {\bf arbitrary
  rotation}. Hence a crucial part of spectral clustering methods is
how to ``round'' these $k$-eigenvectors to a close-by $k$-partition. 

Formally, we study the problem of approximating a $k$-dimensional
linear subspace of $\R^n$ by another subspace which is {$k$-piecewise}
constant: Every vector of this subspace has its coordinates comprised
of at most $k$ distinct values. 
Or equivalently, given a $k$-by-$n$ orthonormal matrix $Y$ of the form
\(Y=[y_1,\ldots, y_n]\)
(think of $Y$ as an embedding of $n$ points in $\R^k$), our problem is
to find a $k$-partition $\curg=\{T_1,\ldots, T_k\}$ so as to minimize
the total variance under any direction:
\begin{align*}
  \min_{\curg} \max_{z \in \R^k: \|z\|_2=1}
  \sum_{S\in \curg}
  \sum_{u\in S}
  \langle z , y_u - c_S \rangle^2.
  %
  %\label{eq:main-goal-0}
\end{align*}
Here $c_S$ is the mean of points in the cluster $S$. If we use
$C \in \R^{n\times k}$ to denote the matrix of cluster means with each
row of $C$ being one of the cluster means, then our objective can be
stated more concisely as \(\min_C \|Y-C\|_2\) with
$\|\cdot\|_2$ being the spectral norm. Geometrically speaking, this
corresponds to finding a $k$-piecewise constant subspace that makes the
{minimum angle} with $Y$.
This is the problem of clustering with spectral norm~\cite{kk10}.
In $2$-dimensions, where $k=2$, optimal solution corresponds to one of
the threshold cuts. From this perspective, our problem can be seen as
a generalization of thresholding to higher dimensions.
%
%This
%problem is at the heart of many spectral methods used in data
%clustering and graph partitioning; which is our main motivation.

\iffalse
The formal definition of our problem is as follows. Given a
$k$-dimensional subspace of $\R^n$ as an orthonormal matrix
$Y \in \R^{n\times k}$, $Y^T Y=I_k$, find a $k$-partition of its rows
such that, for $C \in \R^{n\times k}$ being the matrix of cluster
centers (each row of $C$ is one of the cluster centers), $\|Y - C\|_2$
is as small as possible. Geometrically, this corresponds to a
$k$-piecewise constant subspace that makes the {minimum angle} with
$Y$. This problem of finding a $k$-partition so as to minimize the
spectral norm was first formalized in~\cite{kk10}. 
%
Since the spectral
norm of difference is bounded, all original spectral properties of $Y$
also apply for our $k$-partition up to a small error. Especially for
basic SDP relaxations, this means that such $k$-partitions will be
provably good solutions. %
%
\fi

Our main contribution is a new spectral clustering algorithm that can
recover a $k$-partition whose center matrix $C'$ satisfies
$\|Y-C'\|_2 \le O(\sqrt{\opt})$, where $\opt$ is the minimum possible
(observe that $\opt \le 1$).  Furthermore, the recovered $k$-partition
will be $O(\sqrt{\opt})$-close in Jaccard index to the optimum
partition: 
%in a very strong sense: 
Each cluster we found will be close to a unique cluster
among the optimum $k$-partition.  Previously, no algorithm was known
to find a $k$-partition closer than $o(k \cdot \opt)$.

We also study two closely related problems. In the first one, the goal
is to approximate a matrix in spectral norm by a block diagonal
matrix, with every block being the normalized adjacency matrix of a
clique. In our second application, we turn to the problem of
\minmaxkpar. Given an undirected, weighted graph $G=(V,C)$; find a
$k$-partition $\curg=\{S_1,\ldots, S_k\}$ of the nodes so as to
minimize the maximum expansion:
\begin{align}
  %\phi(\curg) \triangleq 
  \Phi_k(G) \triangleq 
  \min_{\curg} 
  \max_{S\in \curg} \frac{C(S,\overline{S})}{\min(|S|,
  |\overline{S}|)}.   
  \label{eq:minmaxkpar-def-0}
\end{align}
Here $C(S,\overline{S})$ denotes the total weight of edges crossing
$S$.
Our second application is for approximating the optimum $k$-partition
of \minmaxkpar\ on graphs
%
%finding non-expanding sets in a
whose spectrum grows faster than $\phi_k$ (we will make this precise
later).

The choice of spectral norm to measure the closeness of associated
subspaces is quite natural from the perspective of our second
application. Given any subspace, we show how to construct graphs in
polynomial time, such that approximating \minmaxkpar\
%finding $k$ non-expanding cuts 
on such graphs implies a solution for the spectral clustering problem.
% with
%spectral norm on any given subspace. 
From this perspective, we can see
that the subspace rounding problem is a {\bf prerequisite} toward
obtaining a $o(k)$-factor approximation algorithm for the problem of
\minmaxkpar, where the best known is $O(k^4)$ due to \cite{lgt14}.
\subsection{Related Work}
\label{sec:related-work}
Spectral methods have been successfully used for clustering
tasks~\cite{b13} arising in many different areas such as
VLSI~\cite{aky99}, machine learning, data analysis~\cite{njw01} and
computer vision~\cite{sm00,ys03}. They are usually obtained by
formulating the clustering task as a combinatorial optimization
problem (such as sparsest/normalized cuts~\cite{sm00}), then solving
the corresponding basic SDP relaxation, whose solution is often given
by $k$ extremal eigenvectors of an associated matrix.

One of the first spectral clustering algorithms with worst case
guarantees was given in~\cite{kvv04} for the graph partitioning
problem assuming certain conditions on the internal versus external
conductance. The problem of finding a $k$-partition so as to minimize
the spectral norm was first introduced by~\cite{kk10} in the context
of learning mixtures of Gaussians. The best known approximation factor
is $O(k)$ due to~\cite{as12}.
A problem closely related to spectral clustering is \minmaxkpar, as
defined in~\eqref{eq:minmaxkpar-def-0}. 
%Here the goal is to partition
%the nodes of a graph into $k$-parts so as to minimize the maximum
%expansion (ratio of the total weight of edges cut and the total number
%of pairs cut) of any individual part.
%
When all cluster sizes are constrained to be nearly equal, this
problem admits a $O(\sqrt{\log n \log k})$-factor
approximation~\cite{bfkmns11}. On the other hand, if a bi-criteria
approximation is sought, then one can find $(1-\Omega(1)) k$ clusters
each of which has expansion at most $O(\sqrt{\log n \log k})$ times
the optimum \cite{lm14}.

If we look at the basic SDP relaxation of \minmaxkpar, then the
optimal fractional solution is given by the $k$ smallest eigenvectors
of the corresponding graph Laplacian matrix. In fact, this is the main
motivation behind the usage of $k$-eigenvectors for clustering tasks
in practice. %
A natural question is whether one can ``round'' these eigenvectors to
a $k$-partition (the so called Cheeger inequalities). When $k=2$, it
was shown in~\cite{am85} that simple thresholding yields a
$2$-partition with $O(\sqrt{\phi_2})$ expansion, where $\phi_k$ is the
optimal value for \minmaxkpar. Later a better bound was given
in~\cite{kllgt13}, assuming there is some gap between
eigenvalues. When $k > 2$, bi-criteria versions of Cheeger's
inequality are known \cite{abs10,lrtv12,lgt14}. Here the guarantees on
the expansion are of the form $\widetilde{O}(\sqrt{\phi_k})$, where
$\widetilde{O}$ hides the dependencies on logarithmic factors; but the
algorithms can only find $(1-\Omega(1)) k$ parts.
The problem becomes significantly harder when {\bf exactly} $k$
clusters are desired. In this case, it was shown in~\cite{lgt14} that
a method similar to the one proposed in~\cite{njw01} will yield a
$k$-partition with maximum expansion $O(k^4 \sqrt{\phi_k})$. This is
the best known approximation algorithm for \minmaxkpar\ problem and,
as of yet, there is no algorithm known which achieves a
poly-logarithmic approximation.

Perhaps the simplest case of \minmaxkpar\ is when there is a gap
between the $(k+1)^{st}$ smallest eigenvalue of Laplacian matrix,
$\lambda_{k+1}$, and $\phi_k$ of the form
\(\frac{\lambda_{k+1}}{\phi_k} \ge \frac{1}{\eps}.\)
One might think of this as a stability criteria: It implies that all
$k$-partitions with maximum expansion $\le O(\phi_k)$ are
$O(\eps)$-close to each other. To put it in another way, approximating
the optimum $k$-partition is at least as easy as finding a
$k$-partition with minimum possible expansion among all its clusters.
For the case of $k=2$, it is trivial to show that thresholding the
second smallest eigenvector of Laplacian yields $\eps$-close partition
to the optimal one.
On the other hand, when $k>2$, the best prior result is due
to~\cite{as12}, which can find a $k$-partition that is
$O(k \eps)$-close to the optimal one. In other words, when
$\eps \gg \frac1k$, there is no algorithm known to find a non-trivial
approximation of the optimum $k$-partition.
\subsection{Organization}
\label{sec:contributions}
We first introduce some useful notation and background
in~\Cref{sec:background}. After this, we state our main contributions
in~\Cref{sec:contributions}.  Then we propose a new spectral
clustering algorithm in~\Cref{sec:algo}.  In~\Cref{sec:analysis}, we
will prove that our algorithm always finds a $k$-partition that is
$\sqrt{\eps}$-close to any given subspace, where $\eps$ is the
optimum.  In~\Cref{sec:applications}, we discuss some applications of
our algorithm. Our main applications will be:
\begin{itemize}
\item (\Cref{sec:app-minmakpar}) Approximating a graph using disjoint
  union of expanders and,
\item (\Cref{sec:app-apprx}) \minmaxkpar\ when
  $\phi_k \le O(\lambda_{k+1})$.
\end{itemize}
Finally, in \Cref{sec:reduction-ncuts}, we present a simple reduction
from \minmaxkpar\ to our problem: This means any algorithm for
\minmaxkpar\ {has to solve} our subspace rounding problem as well.

\section{Notation and Background}
\label{sec:background}
Let $[m]\triangleq \{1,2,\ldots,m\}$.  We will associate $V=[n]$ with
the set of nodes.
For any vector $q \in \R^\Upsilon$,
$\unit{q} \triangleq \frac{1}{\|q\|_2} q$ and
$\dunit{q} \triangleq \frac{1}{\|q\|^2_2} q$.
Given a subset $S \subseteq V$, we use $\ind{S} \in \R^n$ to denote
the indicator vector for $S$, $\ind{S}(i) =
\begin{dcases*}
  1 & if $i \in S$, \\
  0 & else.
\end{dcases*}$.

\paragraph{Matrices.} We use $\R^{r\times c}$ to denote the set of
$r$-by-$c$ real matrices. Likewise, we use $\symm^c$ and
$\sympm^c \subseteq \symm^c$ to denote the set of $c$-by-$c$ symmetric
and positive semidefinite matrices, respectively. Finally let $\orthk$
be the set of all $n$-by-$k$ orthonormal matrices (Stiefel manifold)
for $k\le n$:
\[
\orthk \triangleq \Big\{ A \in \R^{n\times k} \ \Big|\
  A^T A = I_k
\Big\}.
\]
Given an $r$-by-$c$ matrix $A \in \R^{r\times c}$, we use
$\sigma_i(A)$, $i \in \{1,2,\ldots,\min(r,c)\}$ to refer to the
$i^{th}$ largest singular value of $A$.
We define $\sigma_{\min}(A)$ as the minimum singular value of $A$, 
and $\|A\|_2$ as the $2$-norm of $A$, which is
$\|A\|_2 = \sigma_1(A)$. Likewise $\|A\|_F$ denotes the Frobenius norm
of $A$, $\|A\|_F = \sqrt{A^T A} = \sqrt{\sum_j \sigma_j^2(A)}$.
Given matrix $R\subseteq [r], C \subseteq [c]$, we will use $A_{R,C}$
to refer to the minor corresponding to rows $R$ and columns $C$.

Finally, we will use $A^{\proj}, A^{\perp} \in \sympm^{r}$ to denote
the $r$-by-$r$ projection matrices onto the column space and co-kernel
of $A$, respectively.
Observe that for any $A\in \orthk$,
\(A^\proj = A A^T\) and \(A^\perp = I_n - A A^T\).

One way of measuring the closeness of two subspaces is to look at how
much (in degrees) we need to rotate a vector in one subspace to the
closest vector in the other subspace. It is well known that
this quantity is related to the spectral norm.  For completeness, we
provide a formal version of this statement along with its proof:
\begin{myproposition}[\cite{ss90}]
  \label{thm:subspace-sim}
  Given two linear $k$-dimensional subspaces of $\R^n$ with
  orthonormal basis $A, B\in \orthk$ respectively; the cosine of the
  largest angle between these two subspaces is given by the following:
  \[
  \cos(\angle AB) \triangleq
  \min_{x\in \spn(A)}
  \max_{y\in \spn(B)} \frac{|\langle x , y \rangle|}{\|x\|_2 \|y\|_2}.
  \]
  We have 
  \(
  \sin(\angle AB) = \|A^\perp B\|_2 = \|B^\perp A\|_2.
  \)
\end{myproposition}
\begin{proof}
  From the definition of $\measuredangle AB$, it is easy to see how it
  measures the maximum degrees necessary to rotate a point in $A$ to
  any point in $B$ and vice versa. We will now prove the second
  statement. Any point $x$ in $\spn(A)$ can be written as $A p$ for
  some $p\in \R^k$. Moreover $A$ is orthonormal, thus
  $\|x\| = \|A p\| = \|p\|$. This allows us to rewrite
  $\cos(\measuredangle AB)$ as follows:
  \[
  \min_{x\in \spn(A)} \max_{y\in \spn(B)}
  \frac{|\langle x , y \rangle|}{\|x\|_2 \|y\|_2}
  = \min_{p} \max_{y\in\spn(B)}
  \frac{|\langle A p , y \rangle|}{\|p\|_2 \|y\|_2}.
  \]
  For any $p$, best $y$ is given by $B^{\proj} A p$.  Moreover
  $\|B^{\proj} A p\|^2_2 + \|B^{\perp} A p\|^2_2 = \|p\|^2_2$, thus:
  \[
  = \min_{p} 
  \frac{\|B^\proj A p\|_2}{\|p\|_2} = \sqrt{1 - \max_p
    \frac{\|B^\perp A p\|^2_2}{\|p\|^2_2}}.
  \]
  Consequently,
  \(
  \sin(\angle AB) = \max_p \frac{\|B^\perp A p\|^2_2}{\|p\|^2_2}
  = \|B^\perp A\|^2_2. 
  \)
\end{proof}

\begin{definition}
  Let $\setv(k)$ be the family of sets of $k$ non-empty subsets of
  $V$.  We will use $\disjv(k) \subseteq \setv(k)$ to denote the set
  of $k$ disjoint subsets of $V$: %
  $\curg \in \disjv(k)$ if and only if $\curg\in\setv(k)$ and
  $S\cap T = \es$ for all $S\neq T \in \curg$.
\end{definition}
In order to compare subspaces with $k$-partitions, we need to identify
a canonical representation of the subspaces associated with
$k$-partitions. The most natural representation is to use each basis
vector as the normalized indicator of one of the clusters.
\begin{notation}[Basis Matrices of $k$-Partitions]
  \label{def:kpar-mat}
  Given $k$-subsets $\curg=\{A_1,\ldots, A_k\}$ of $V$, let
  $\mcurg \in \R^{n\times k}$ be the corresponding normalized
  incidence matrix
  $\mcurg \triangleq \big[ \unit{\ind{A_1}} \quad \ldots\quad
  \unit{\ind{A_k}} \big]$.
  We will use $\mcurg^\proj\in \sympm^n$ and
  $\mcurg^\perp\in \sympm^n$ to denote the associated projection
  matrices so that $\mcurg^\proj \mcurg = \mcurg$ and
  $\mcurg^\perp \mcurg = 0$.
\end{notation}
Multiplication with either of the projection matrices $\mcurg^\proj$
and $\mcurg^\perp$ has a natural correspondence with means and the
differences to means:
\begin{myproposition}
  \label{thm:gamma-basic-props}
  If \(\curg\in\disjv\),
  then $\mcurg$ is an orthonormal matrix, \(\mcurg \in \orthk\),
  and \(\mcurg^\perp\) is a Laplacian matrix.
  For any \(Y \in \R^{k\times n}\), $i^{th}$ column of:
  \begin{itemize}
  \item %
    \(Y \mcurg^\proj\)
    is the mean of points in the same cluster with
    $i$ provided $i$ is in any cluster of $\curg$,
    and $0$ otherwise.
  \item %
    \(Y \mcurg^\perp\)
    is the difference between $y_i$
    and its associated center as defined above.
  \end{itemize}
  For example, \(\|Y \mcurg^\perp\|^2_F\) measures the sum of
  squared distances of each point to the center of its cluster or
  origin if it is not in any cluster.
\end{myproposition}
We will measure the distance between sets in a way similar to cosine
distance.
\begin{notation}%
  Given $p, q\in \R^n$, we define $\simfun(p,q)$ as
  $\simfun(p,q)\triangleq 1 - \langle \unit{p} , \unit{q}\rangle^2$.
  Note that $\simfun(p,q)
  =\frac12 \| {\unit{p}}^{\otimes 2} -  {\unit{q}}^{\otimes 2}\|^2
  $.
  For convenience, we will use $\simfun(S,q)$ as
  $\simfun(\ind{S}, q)$. In particular,
  $\simfun(A,B) = 1 - \frac{|A\cap B|^2}{{|A| |B|}}$.
\end{notation}
Our measure of set similarity is closely related to the Jaccard index.
\begin{myproposition}
  \label{thm:simfun-props}
  For any pair of subsets $A, B\subseteq V$:
  \[
    \frac14 \frac{|A\Delta B|}{|A\cup B|}
    \le
     \simfun(A,B) \le \frac{|A\Delta B|}{|A\cup B|}.
  \]
\end{myproposition}
\begin{proof}
  Since $|A\cup B|^2 \ge {|A| |B|}$, we immediately see that
  $1-\simfun(A,B) \ge \frac{|A\cap B|}{|A\cup B|}$.  For the other
  direction, suppose $\simfun(A,B) \le \eps$ and $|A| \ge |B|$. Then
  \(
  (1-\eps) \sqrt{|A| |B|} \le |A\cap B|
  \)
  which implies
  \[
  \frac{|A\cap B|}{|A|} \ge 
  (1-\eps) \sqrt{\frac{|B|}{|A|}} \ge
    (1-\eps) \sqrt{\frac{|A\cap B|}{|A|}}.
  \]
  Therefore $|A\cap B| \ge (1-\eps)^2 |A|$
  and $|A\Delta B| = |A| + |B| - 2 |A\cap B|
  \le |B| - (1-4 \eps + 2\eps^2)|A|
  \le 4 \eps |B|$. In particular,
  \[
    \frac{|A\Delta B|}{|A\cup B|}
    \ge 4 \simfun(A,B).
%    \tag*{\qedhere}
  \]
\end{proof}
We will now generalize our set similarity measure to $k$-partitions.
\begin{notation}%
  \label{def:kpar-sim}
  Given $\curg, \newg \in \setv(k)$; we define
  $\simfun(\curg,\newg)$ as:
  \[
  \simfun(\curg,\newg) \triangleq \min_{\pi: \curg\leftrightarrow
    \newg} \max_{S\in \curg} \simfun(S, \pi(S)).
  \]
  We say $A$ (resp. $\curg$) is $\eps$-close to $B$ (resp. $\newg$)
  whenever $\simfun(A,B) \le \eps$ (resp.
  $\simfun(\curg,\newg) \le \eps$).
\end{notation}
Observe that our notion of proximity is a very strong bound. For
example if $\curg$ is $\eps$-close to $\optg$, then any subset
$S \in \optg$ of size $|S| < \frac{1}{\eps}$ has to be preserved
{\em exactly} in $\curg$.

The next theorem says that the similarity measure we use for
$k$-partitions in Notation~\ref{def:kpar-sim} is tightly related to the
spectral norm distance between the corresponding basis.
\begin{mytheorem}
  \label{thm:set-similarity}
  Given $\curg, \newg\in \disjv(k)$;
  $
    \simfun(\curg,\newg) \le 
    \| \mcurg^\perp \mnewg \|^2_2 \le 2 \simfun(\curg,\newg).
  $
  Moreover, after appropriately ordering the columns of
  $\mnewg$, $\|\mcurg - \mnewg\|_2^2 \le 4 \simfun(\curg,\newg)$.
\end{mytheorem}
Proof of Theorem~\ref{thm:set-similarity} is given
in~\Cref{sec:proof-of-set-similarity}.
%
%
%
\iffalse
\begin{myproposition}
  \label{thm:small-dif-sigma-k}
  Given two symmetric matrices $A, B \in \symm^m$; if $A \preceq B$
  then $\sigma_r(A) \le \sigma_r(B)$ for any $r$. 
\end{myproposition}
%
\begin{proof}
  Recall that
  $\sigma_r(M) = \max_{P: P^T P = I_r} \sigma_{\min}(P^T M P)$ for any
  symmetric matrix $M$. Therefore:
  \[
    \sigma_r(A) = \max_{P: P^T P = I_r} \sigma_{\min}(P^T A P)
    \le \max_{P: P^T P = I_r} \sigma_{\min}(P^T B P)
    = \sigma_r(B). %\tag*{\qedhere}
  \]
\end{proof}
\fi
%

\begin{myproposition}
  \label{thm:sing-max-to-min}
  Given $A, B \in \orthk$, %
  $
     \sigma_{\min}(A^T B) = \sqrt{1 - \| A^\perp B\|_2^2}. 
  $
\end{myproposition}
\begin{proof}
  $B^T A^\perp B = B^T B - B^T A A^T B = I_k - B^T A A^T B$.  Since
  $\|A^T B\|_2 \le 1$, $\|A^\perp B\|_2^2 = 1 - \sigma_{\min}(A^T B)^2$.
\end{proof}
Consider two subspaces with basis $A$
and $B$. If the angle between these two subspaces is small, then one might
intuitively expect that $A A^T$ and $B B^T$ are very close also.
In the next lemma, we make this intuition formal. We also include its
proof for completeness.
\begin{mylemma}[\cite{ss90}]
  \label{thm:spec-rad-bnd}
  \label{thm:sandwich}    
  Given $A, B \in \orthk$; %
  $\| A A^T - B B^T \|_2 =
  %\le 2 \sqrt{2} 
  \|A^\perp B\|_2$.
\end{mylemma}
\begin{proof} 
%\end{proof}
%\begin{proof}
  %Let $\sigma\triangleq \| A A^T - B B^T \|_2$. Since $\sigma^2 = \|
  %(A A^T - B B^T)^2 \|_2$, 
  We will prove $\le$ by 
  %We will 
  upper bounding the spectral norm of $(A A^T - B B^T)^2$. Since
  $A A^T= A^{\proj}$ and $B B^T = B^{\proj}$:
  \begin{align}
    (A A^T - B B^T)^2
    = & A^{\proj} + B^{\proj} 
    - A^{\proj} B^{\proj} - B^{\proj} A^{\proj} \notag \\
    = & A^{\proj} B^{\perp} + B^{\proj} A^{\perp}.
    \label{eq:917828}
  \end{align}
  If this matrix is zero, then our claim is trivially true. Suppose
  not. Consider the largest eigenvalue $\sigma$ of \cref{eq:917828}
  and a corresponding eigenvector $q$.
  We have
  $0 \neq \sigma q = (A^{\proj} B^{\perp} + B^{\proj} A^{\perp}) q$,
  which means either $B^{\perp} q\neq 0$ or $A^{\perp} q \neq 0$ (or
  both). Without loss of generality, we may assume
  $A^{\perp} q \neq 0$:
  \[
  \sigma q = (A^{\proj} B^{\perp} + B^{\proj} A^{\perp}) q \implies
  \sigma A^{\perp} q = A^{\perp} B^{\proj} A^{\perp} q.
  \]
  Consequently, $q'\triangleq A^{\perp} q$ is an eigenvector of
  $A^{\perp} B^{\proj} A^{\perp}$ with eigenvalue $\sigma$:
  \[
  A^{\perp} B^{\proj} A^{\perp} q'
  = A^{\perp} B^{\proj} A^{\perp} q = \sigma A^{\perp} q = \sigma q'.
  \]
  In particular, 
  $\|A A^T - B B^T\|^2_2 =
  \|(A A^T - B B^T)^2\|_2 = \sigma \le \|A^{\perp} B^{\proj}
  A^{\perp}\|_2 = \|A^\perp B B^T A^{\perp}\|_2 = \|A^\perp B\|^2$.
  %
  % Consider an SVD of the matrix $A^T B = C \Sigma^{1/2} D^T$.  Here
  % $C,\Sigma,D \in \R^{k\times k}$.  For $R \triangleq C D^T$,
  % $R^T R = R R^T = I_k$. Note
  % $\|R - A^T B\|_2 = \|I-\Sigma^{1/2}\|_2= 1-\sigma_{\min}(A^T B) =
  % \rho$
  % and $\|A^\perp B\|_2^2 = \|B^\perp A\|_2^2 = 1-(1-\rho)^2$
  % by~\Cref{thm:sing-max-to-min}.  Using the fact that $A^\perp A = 0$:
  % \begin{align*}
  %   \| A R - B\|_2^2 = & \| A (R - A^T B) - A^\perp B\|_2^2 \\
  %   = & \|A(R-A^T B)\|_2^2 + \| A^\perp B\|_2^2 =
  %       \rho^2 + 1 - (1-\rho)^2=
  %       2 \rho. 
  % \end{align*}
  % %
  % Thus for any $q$:
  % %
  % \begin{align*}    
  %   \Big| \|A^T q\|_2 - \|B^T q\|_2\Big|
  %   = &
  %       \Big| \|(A R)^T q\|_2 - \|B^T q\|_2\Big| \\
  %   \le& \Big| \|B^T q\|_2 - \|B^T q\|_2\Big| 
  %        + \|(A R - B)^T q\|_2 \\
  %   = & \|(A R - B)^T q\|_2 \le \sqrt{2 \rho} \|q\|_2.  
  % \end{align*}
  % %
  % Consequently:
  % %
  % %
  % \[
  % | q^T (A A^T - B B^T) q|
  % = \big| \|A^T q\|_2 - \|B^T q\|_2 \big| \big( \|A^T q\|_2 + \|B^T q\|_2\big)
  % \le 2 \sqrt{2 \rho} \|q\|_2^2.
  % \]
  % We finish the proof by noticing that
  % \({\rho} = {1-\sigma_{\min}}={1-\sqrt{1-\|A^\perp B\|^2_2}} \le
  % \|A^\perp B\|_2^2\).
  %
  Now we will prove $\ge$. If we multiply both sides of
  \cref{eq:917828} with $A^{\perp}$, we see that
  \[
  (A A^T - B B^T)^2 \succeq 
  A^\perp (A A^T - B B^T)^2 A^{\perp} 
  = A^{\perp} B^{\proj} A^{\perp}
  \]
  which implies 
  \[
  \|(A A^T - B B^T)^2\|_2
  \ge \|A^{\perp} B\|^2_2. %\tag*{\qedhere}
  \]
\end{proof}

\subsection{Graph Partitioning}
\label{sec:back:graph-par}
Given an undirected graph $G=(V,C)$
with nodes $V$
and non-negative edge weights $C$,
we use \(A_G \in \symm^V\)
and \(L_G \in \sympm^V\)
to denote the adjacency and Laplacian matrices of $G$.
Consider the following $k$-way
graph partitioning problem where the goal is to minimize the maximum
ratio of the total weight of edges cut and the number of nodes inside
among all clusters.
\begin{definition}
  [\minmaxkpar]
  \label{def:minmaxkpar}
  Given an undirected graph $G=(V,C)$
  with nodes $V$
  and non-negative edge weights $C$,
  we define the $k$-way expansion of $G$ as the following:
  \[
  \phi_k(G) \triangleq \min_{\curg \in \disjv(k)} \max_{T\in \curg}
  \frac{C(T,\overline{T})}{|T|}.
  \]
  Here \(C(A,B)\) denotes the total weight of unordered edges
  between \(A\) and \(B\). For fixed \(G\), we will use
  \(\optg \in \disjv(k)\) to refer to the $k$-partition which
  achieves \(\phi_k(G)\).
\end{definition}
At the first glance, our notion of expansion might seem different than
the usual definition given in~\cref{eq:minmaxkpar-def-0}. However they
are indeed the same:
\begin{myproposition}
  \label{thm:expansion-equiv}
  For any $G=(V,C)$ and a $k$-partition of $V$, $\curg \in \disjv(k)$,
  \[
  \max_{T\in \curg} \frac{C(T,\overline{T})}{|T|}
  = 
  \max_{S\in \curg}
  \frac{C(S,\overline{S})}{\min(|S|, |\overline{S}|)}.
  \]
  In particular, $\phi_k(G) = \Phi_k(G)$. 
\end{myproposition}
\begin{proof}
  For any $S$, $\min(|S|, |\overline{S}|) \le |S|$, therefore
  $\frac{C(S,\overline{S})}{|S|} \le
  \frac{C(S,\overline{S})}{\min(|S|, |\overline{S}|)}$.
  Now we will prove the other direction. Let
  $\phi = \max_{T\in \curg} \frac{C(T,\overline{T})}{|T|}$. 
  For any $T' \in \curg$:
  \[
  C(T', \overline{T'}) \le 
  \sum_{T \in \curg\setminus T'} C(T,\overline{T})
  \le \phi \sum_{T\in \curg\setminus T'} |T|
  = \phi |V\setminus T'|.
  \]
  Consequently,
  \( \frac{C(T', \overline{T'})}{|V\setminus T'|} 
  = \frac{C(T', \overline{T'})}{|\overline{T'}|} 
  \le \phi.  \)
  Recall that
  $\frac{C(T', \overline{T'})}{|T'|} \le \phi$, so \( 
  \frac{C(T', \overline{T'})}{\min(|T'|, |\overline{T'}|} \le \phi\).
  % \[
  % \forall T\in \curg: \phi \ge \max\left(\frac{C(T,\overline{T})}{|T|},
  %   \frac{C(T,\overline{T})}{|\overline{T}|}\right)
  % = \frac{C(T,\overline{T})}{\min(|T|, |\overline{T}|}. %\tag*{\qedhere}
  % \]
\end{proof}
We can capture the objective function of $\phi_k$ using the spectral
norm, within a factor of $2$:
%The main benefit of using $\phi_k$ is that, w
%
\begin{mylemma}
  \label{thm:minmaxkpar-as-snorm}
  Given $\curg\in \disjv(k)$,
  \(
  \frac12 \| \mcurg^T L \mcurg \|_2
  \le 
  \max_{T\in \curg} \frac{C(T,\overline{T})}{|T|}
  \le
  \| \mcurg^T L \mcurg \|_2.
  \)
\end{mylemma}
\begin{proof}
  Let
  \(\phi\triangleq \max_{T\in \curg} \frac{C(T,\overline{T})}{|T|}\).
  We need to prove
  \( \phi \le \sigma_{\max} (\mcurg^T L \mcurg) \le 2 \phi.  \)
  The lower bound is trivial, so we only give the proof of upper
  bound.
  Note $\mcurg = J D^{-1/2}$ where $D$ is a matrix whose diagonals are
  $|T|$ for $T\in \curg$ and the columns of $J$ are indicator vectors
  for every $T\in \curg$. Then
  \( \mcurg^T L \mcurg = D^{-1/2} J^T L J D^{-1/2}.  \)
  Define $W$ as the matrix which is equal to $J^T L J$ along its
  diagonals and $0$ everywhere else. Since $J^T L J$ is a Laplacian
  matrix $ J^T L J \preceq 2 W $. Therefore:
  \[
  \mcurg^T L \mcurg \preceq 2 D^{-1/2} W D^{-1/2}.
  \]
  \( D^{-1/2} W D^{-1/2}\)
  is diagonal whose entries are %
  \( \frac{\ind{T}^T L \ind{T}}{|T|} = \frac{C(T,\overline{T})}{|T|}
  \le \phi \) over all $T\in\curg$.  Consequently,
  \[
  \sigma_{\max}(\mcurg^T L \mcurg) \le 2 \sigma_{\max}( D^{-1/2} Z D^{-1/2} )
  \le 2 \phi.  %\tag*{\qedhere}
  \]
\end{proof}
Given Lemma~\ref{thm:minmaxkpar-as-snorm}, a simple relaxation for
$\phi_k(G)$ (the basic SDP relaxation) is the following:
\begin{align}
\min{\| Q^T L Q \|_2}\ \st\ Q^T Q = I_k. \label{eq:minmaxkpar-sdp}
\end{align}
Note that for any $\curg\in\disjv(k)$, $\mcurg$ is feasible and
\cref{eq:minmaxkpar-sdp} is indeed a relaxation. Moreover the
Courant-Fischer-Weyl principle implies that the optimum value of
\cref{eq:minmaxkpar-sdp} is $\lambda_k$ with the corresponding optimal
solution being the smallest $k$-eigenvectors of $L$. Therefore:
\begin{align}  
  \lambda_k \le \phi_k(G).
  \label{eq:lambda-phi}
\end{align}
\section{Our Contributions}
\label{sec:contributions}
%
%We study the spectral clustering problem for orthonormal matrices as
%in~\cref{eq:main-goal-0}.
%Recall our main
%problem from~\cref{eq:main-goal-0}: 
We re-state our main problem. 
Given a $k$-by-$n$ matrix
$Y: Y^T \in \orthk$ of the form \(Y=[y_1,\ldots, y_n]\)
(think of $Y$ as an embedding of $n$ points in $\R^k$), find a
$k$-partition $\curg\in \disjv(k)$ so as to minimize the total
variance under any direction:
\begin{align}
  \min_{\curg} \max_{z \in \R^k: \|z\|_2=1}
  \sum_{S\in \curg}
  \sum_{u\in S}
  \langle z , y_u - c_S \rangle^2
  .\label{eq:main-goal-1}
\end{align}
Here $c_u$ is the mean of points in the same cluster with $u$ if one
exists, and $c_u = 0$ otherwise. This is the problem of clustering
with spectral norm~\cite{kk10}.
\begin{remark}[Covering all points]
  For simplicity, we allow some points to be left uncovered by any set
  in $\curg$. 
  %, which corresponds to their ``center'' being at the
  %origin. 
  However the same guarantees still hold even if we require $\curg$ to
  cover all points: We arbitrarily assign uncovered points to clusters 
  while making sure that the relative cluster sizes do not change. 
  This procedure changes the approximation ratio by a factor of $2$. 
\end{remark}
We can express \cref{eq:main-goal-1} more succinctly as the following:
\begin{align}
  \min_{\curg} \| Y \mcurg^\perp \|_2^2
  \overset{\textrm{Proposition~\ref{thm:gamma-basic-props}}}{=}
  \cref{eq:main-goal-1}
  \label{eq:main-goal}
\end{align}
There are two closely related problems, whose optimum is within square
root of \cref{eq:main-goal-1} (Lemma~\ref{thm:sandwich}): 
\begin{itemize}
\item Finding a $k$-by-$k$ rotation matrix $R: R^T R = R R^T = I_k$
  and a $k$-partition $\curg\in \disjv(k)$ so as to minimize the
  following:
  \begin{align}
    \min_{R, \curg} \| R Y - \mcurg \|_2. \label{eq:main-goal-r}
  \end{align}
\item Approximate the Gram matrix of $Y$, $Y^T Y$, using block
  diagonal matrices with each block being constant. This is equivalent
  to:
  \begin{align}
    \min_{\curg} \| Y^T Y - \mcurg \mcurg^T \|_2. \label{eq:main-goal-r}
  \end{align}
\end{itemize}
Our main contribution is a new spectral clustering algorithm whose
pseudo-code is given through
\Cref{alg:unravel,alg:find-cluster,alg:boost,alg:round-vec,alg:spectral-clustering}.
We prove the following guarantee on its outputs.
\begin{mytheorem}[Restatement of Theorem~\ref{thm:spectral-clustering}]
  Let $\optg \in \disjv(k)$ with $\|Y \moptg^\perp\|^2_2 \le O(\eps)$.
  Then $\newg \gets \textsc{SpectralClustering}(Y)$ is a $k$-partition
  so that $\newg \in \disjv(k)$, and it is $O(\sqrt{\eps})$ close to
  both $\optg$ and $Y$:
  \[
  \simfun(\optg, \newg) \le O( \sqrt{\eps} )
  \quad\mbox{and}\quad
  \| Y \mnewg^\perp \|^2_2 \le O(\sqrt{\eps}).
  \]
\end{mytheorem}
\begin{remark}[Small Clusters]
  Our main guarantee as stated in Theorem~\ref{thm:spectral-clustering} works
  for any cluster size. For example, consider the case of some optimal
  cluster $T\in \optg$ having size $|T| \le O\big(1/\sqrt{\eps}\big)$.
  For such cluster, any $S$ with $\simfun(S,T) \le O(\sqrt{\eps})$ has
  to be exactly equal to $T$.

  In other words, our algorithm will recover any $T \in \optg$ with
  $|T| \le O(1/\sqrt{\eps})$ {\bf exactly}.
\end{remark}
As an easy consequence, we show how to approximate a graph as a
disjoint union of expanders (provided one exists) in polynomial time.
\begin{mycorollary}[Restatement of Corollary~\ref{thm:app-graph-par}]
  Given a graph $G$, if there exists $\optg \in \disjv(k)$ such that
  Laplacian of $G$ is $\eps$-close (in spectral norm) to the Laplacian
  corresponding to the disjoint union of normalized cliques on each
  $T \in \optg$:
  \[
    \| L - \moptg^\perp \|_2 \le \eps, 
  \]
  then in polynomial time, we can find $\curg\in\disjv(k)$ which is
  $O\big(\sqrt{\eps}\big)$-close to $\optg$ and $G$:
  \[
    \big\| L - \mcurg^\perp \|_2 \le O({\eps}^{1/4}).
  \]
\end{mycorollary}
Next we significantly improve the known bounds for recovering a
$k$-partition when all clusters have small expansion as
in Definition~\ref{def:minmaxkpar}.
Previous spectral clustering algorithms only guarantee recovering each
$T \in \optg$ when the ${(k+1)}^{st}$ smallest eigenvalue,
$\lambda_{k+1}$, of the associated Laplacian matrix for $G$ satisfies
\[\lambda_{k+1} > \Omega(k \cdot \phi_k).\] Our new algorithm significantly
relaxes this requirement to $\lambda_{k+1} > \Omega(\phi_k)$.
\begin{mytheorem}[Restatement of Theorem~\ref{thm:app-min-max-exp}]
  Given a graph $G$ with Laplacian matrix $L$, let $\curg$ be the
  $k$-partition obtained by running~\Cref{alg:spectral-clustering} on
  the smallest $k$ eigenvectors of $L$.
  Then:
  \[
    \Delta(\curg, \optg) \le O\Bigg(\sqrt{\frac{\phi_k}{\lambda_{k+1}}}\Bigg).
  \]
\end{mytheorem}
Finally, we show that any approximation algorithm for \minmaxkpar\
implies the same approximation bound for the spectral clustering
problem restricted to orthonormal matrices. In other words, the
spectral clustering problem is a {\em prerequisite} for approximating
\minmaxkpar\ {\em even} on graphs whose normalized Laplacian matrix
has its $(k+1)^{st}$ eigenvalue $\lambda_{k+1}$ larger than a
constant.
\begin{mytheorem}[Restatement of Theorem~\ref{thm:y-to-x}]
  Given \(Y: Y^T \in \orthk\),
  let \(\optg \triangleq \argmin_{\curg} \|Y \curg^\perp\|_2^2\)
  with $\eps = \|Y \optg^\perp\|^2_2$. Then there exists a weighted,
  undirected, regular graph $X$, whose normalized Laplacian matrix has
  its $(k+1)^{st}$ smallest eigenvalue
  at least \( \lambda_{k+1} \ge 1 - O(\sqrt{\eps}) \) such that:
  \begin{itemize}
  \item Each $T\in \optg$ has small expansion, $\phi_X(T) \le O(\eps)$;
  \item If $\curg \in \disjv(k)$ is a $k$-partition with
    $\max_{S\in \curg} \phi_X(S) \le \delta$, then
    $
      \simfun(\curg, \optg) \le O(\delta + \sqrt{\eps}).
    $
  \end{itemize}
  Moreover such $X$ can be constructed in polynomial time.
\end{mytheorem}
\section{Our Algorithm}
\label{sec:algo}
The pseudo-code of our clustering algorithm and its sub-procedures are
listed
in~\Cref{alg:unravel,alg:find-cluster,alg:boost,alg:round-vec,alg:spectral-clustering}.
The main procedure is invoked by
$\newg=\textsc{SpectralClustering}(Y)$
(\Cref{alg:spectral-clustering}), where $Y$ is a $k$-by-$n$ matrix $Y$
such as the smallest $k$-eigenvectors of some Laplacian matrix. The
output $\newg$ is a $k$-partition close to $Y$. We use
$\optg=(T_1,\ldots,T_k)$ to denote the closest $k$-partition to $Y$.
We will refer to $T$'s as {\em true}  clusters.

\subsection{Intuition}
\label{sec:algo:int}
First we start with the discussion of some of the main challenges
involved in spectral clustering and the intuition behind the major
components of our algorithm.
\vspace{1ex}
\paragraph{Finding a Cluster.}
Since there are $k$ directions and $k$ clusters, we can think of each
direction being associated with one of the clusters. Moreover, for one
of the true clusters $T$, the {\em total} correlation of its center
with remaining directions will be very small. By utilizing this
intuition, we can easily find such a subset, say
$S$~(\Cref{alg:find-cluster}). However this property need not be true
for all $T$'s: Even though each $T$ will be at most $\eps$-close to
every non-associated direction, the {\em total} correlation might be
$(k-1) \eps \gg 1$, thus overwhelming the correlation it had with its
associated direction. In fact, this is the reason why $k$-means type
procedures will fail to find every cluster when $\eps > 1/k$. To
remedy this, each time we find $S$, we can try to ``peel'' it off. A
natural approach is to project the columns of $Y$ onto the orthogonal
complement of the center of $S$. Similar ideas were used before in the
context of learning mixtures of anisotropic Gaussians~\cite{bv08} and
column based matrix reconstruction~\cite{dr10,gs12-col}. After
projection, we obtain a new $(k-1)$-by-$n$ orthonormal matrix, $Z$,
corresponding to remaining $k-1$ clusters.
\vspace{1ex}
\paragraph{Boosting.}
Unfortunately we can not iterate the above approach: No matter how
accurate we are in $S$, there will be some error: If $Y$ were
$\eps$-close to $(T_1,\ldots, T_k)$, then we can at best guarantee
that $Z$ is $2 \eps$-close to $(T_2,\ldots, T_k)$.
After $k$ iterations, our error will be $2^{O(k)} \eps$, which is much
worse than $k$-means! In our algorithm, we keep the error from
accumulating via a boosting step~(\Cref{alg:boost}). 
\vspace{1ex}
\paragraph{Unraveling.}
Even with boosting, there remains one issue: The clusters we found may
overlap with one another. Unlike other distance based clustering
problems such as $k$-means, the assignment problem (``which cluster
does this node belong to?'') is quite non-trivial in spectral
clustering {\em even} if we are given all $k$-centers. There is no
simple {\em local} procedure which can figure out the assignment of
node $u$ by only looking at $y_u$ and the cluster centers.
We deal with this issue by reducing the ownership problem to finding a
matching in a bipartite graph~(\Cref{alg:unravel}). Our approach is
very similar to the one used in~\cite{bs06} for a special case of
Santa Claus problem.
Unfortunately, this operation has a {\em cascading effect}: Adding a
new cluster might considerably change the previous clusters and their
centers. Dealing with this challenge is what causes our final
algorithm to be rather involved.
\vspace{1ex}
\paragraph{Final Algorithm.}
The final difficulty we face is that, the boosted cluster might
cannibalize other much smaller clusters. We overcome this issue by
maintaining both estimates for every true cluster: a coarse estimate,
which is the core we originally found; and the finer estimate
obtained after boosting. Due to the cascading effect of
unraveling whenever we add a new cluster, we have to re-compute the
centers and project onto their orthogonal complement at every round. 
\subsection{Overview} %
\label{sec:algo:overview}
Our algorithm proceeds iteratively.  At $r^{th}$ iteration, it finds a
{\em core} $S_r$ for one of the yet-unseen true clusters, say
$T_r$. The main invariant we need from the core set is the
following (Lemma~\ref{thm:find-cluster-correct}):
\begin{enumerate}[(i)]
\item \label{item:alg:overview:1}$S_r$ is noticeably close to $T_r$, say 
  \(\simfun(S_r,T_r) \le \frac{1}{100}.\)
\item \label{item:alg:overview:2} All the remaining $k-r$ true clusters,
  \(\{T_{r+1}, \ldots, T_k\}\),
  have small overlap with $S_r$ in the sense that
  \(|T_{j} \cap S_r| \le \frac{1}{100} |T_{j}|\)
  over all $j> r$.
\end{enumerate}
After having found $S_r$, our algorithm needs to boost $S_r$ to
$\widehat{S_r}$. For boosting to work however (Theorem~\ref{thm:boost-main}),
we need the invariant~\eqref{item:alg:overview:2} to be true for all
$j\neq i$. We do this by
$(U_1,\ldots, U_r)\gets \textsc{Unravel}(S_1,\ldots, S_r)$. Since
$U_i$'s are close to $S_i$'s, each $T_j$ with $j<r$ will still be
mostly overlapping with $S_j$~(Theorem~\ref{thm:unravel-correct}); hence
$T_j$'s with $j < r$ can not overlap with $U_r$ as $U_j$'s are
disjoint. Consequently we can use $U_r$ instead of $S_r$ for boosting
so as to obtain $\widehat{S_r}$. The only invariant we require from
$\widehat{S_r}$ is that it is much {\em closer} to $T_r$ than $S_r$:
\[
  \simfun(\widehat{S_r}, T_r) \le O(\sqrt{\eps}).
\]
By using the centers of boosted sets instead of core sets, we can make
sure that the error does not accumulate after the projection step
$Z \gets (Y \mnewg')^\perp Y$~(Lemma~\ref{thm:progress}).

After $k$ iterations, we apply $\textsc{Unravel}$ to all boosted sets
$(\widehat{S_1},\ldots, \widehat{S_k})$ one last time and output the
result. 

\begin{remark}
  %[Computing Top Singular Vectors by Power
  %Method~\cite{gkb13}]
  In \Cref{alg:find-cluster,alg:boost}, the last step involves
  computing the top singular vector. In both cases, we have:
  \begin{itemize}
  \item A good initial guess (the indicator vector),
  \item Large separation between $\sigma_1$ and $\sigma_2$.
  \end{itemize}
  Thus, we can simply use power method for $O(\log 1/\eps)$ many
  iterations to compute a sufficiently accurate approximation of the
  top right singular vector, from which we can obtain an approximation
  of the top left singular vector easily. It was previously shown
  in~\cite{gkb13} that power method is sufficient in the context of
  spectral clustering.
\end{remark}
\setlength{\leftmargini}{1em}
\setlength{\leftmarginii}{1.2em}

\begin{program*}[!h]
  \begin{floatrow}[2]
    \floatbox[\captop]{program}[\FBwidth]
    {
      \caption{$\widehat{S} =\textsc{Boost}(Y, S)$.}
      \label{alg:boost}}
    {
      \begin{minipage}[t]{\linewidth}
        \begin{compactenum}
        \item $p \gets$ top left singular vector of $Y_S$.
        \item Return $\textsc{Round}(Y^T p)$.
        \end{compactenum}
      \end{minipage}
    }
    \floatbox[\captop]{program}[\FBwidth]
    {
      \caption{$S=\textsc{Round}(q)$.}
      \label{alg:round-vec}}
    {
      \begin{minipage}[t]{\linewidth}
        \begin{compactenum}
        \item
          $F \gets \big\{  
          \{ u \mid s q_u \ge s q_v \} \mid v \in V, s\in\{\pm 1\}  \}
          $.
        \item Return $\argmax_{S\in F} | \langle q , \unit{\ind{S}} \rangle |$.
        \end{compactenum}
      \end{minipage}
    }   
  \end{floatrow}
  \vspace{0.5ex}
  \begin{floatrow}[2]
    \floatbox[\captop]{program}[\FBwidth]
    {
      \caption{$\newg=\textsc{SpectralClustering}(Y)$.}
      \label{alg:spectral-clustering}}
    {
      \begin{minipage}[t]{\linewidth}
        \begin{compactenum}
        \item For $r\gets 1$ to $\rank(Y)$ do:
          \begin{compactenum} 
          \item
            $\newg' \gets \textsc{Unravel}(\widehat{S_1}, \widehat{S_2},
            \ldots, \widehat{S_{r-1}})$.
          \item $Z \gets (Y \mnewg')^\perp Y$.
          \item $S_r \gets \textsc{FindCluster}( Z )$.
          \item $(U_1,..., U_{r}) \gets \textsc{Unravel}(S_1,..., S_r)$.
          \item $\widehat{S_r} \gets \textsc{Boost}(Y, U_r)$.
          \end{compactenum}
        \item Return $\textsc{Unravel}(\newg)$.
        \end{compactenum}
      \end{minipage}
    }

    \floatbox[\captop]{program}[\FBwidth]
    {
      \caption{$S=\textsc{FindCluster}(Y)$.}      
      \label{alg:find-cluster}}
    {
      \begin{minipage}[t]{\linewidth}
        \begin{compactenum}
        \item Choose $\delta' \in [0,1]$ as the minimum %
          for which the
          following returns some $S$.
        \item For each $c \in V$:
          \begin{compactenum}
          \item $\pi$ be an ordering 
            st
            \(
            \frac{ \| Y_{\pi(i)} - Y_c \| }{ \|Y_{\pi(i)}\| }\nobreak
            \le\nobreak
            \frac{ \| Y_{\pi(i+1)} - Y_c \| }{ \|Y_{\pi(i+1)}\| }.
            \)
          \item
            $m \gets \min \Big\{ j \Big| \sum_{i\le j} \|Y_{\pi(i)}\|^2
            \ge 1 - \delta' \Big\}$.
          \item If $\sum_{j\le m} \|Y_{\pi(j)} - Y_c\|^2 \le \delta'$, then:
            \begin{compactenum}
            \item $S \gets \{\pi(1),\ldots, \pi(m)\}$.
            \item $q \gets$ top right singular vector of $Y_S$.
            \item Return $\textsc{Round}(q)$.
            \end{compactenum}
          \end{compactenum}
        \end{compactenum}        
      \end{minipage}
    }
    
  \end{floatrow}
  \begin{floatrow}[1]
    \floatbox[\captop]{program}[\FBwidth]
    {      
      \caption{$\newg=\textsc{Unravel}_{\delta}(\curg)$ (see~\Cref{fig:unravel:ex}
        for a sample graph construction.)}
      \label{alg:unravel}}
    {
      \begin{minipage}[t]{0.89\linewidth}
        \begin{compactenum}
        \item Choose $\delta' \in [0,1]$ as the minimum for which the
          following returns a matching. %
        \item Construct a bipartite graph $H=(L,R,E)$,
          where left side $L$ is $V$: %
          \begin{compactenum}
          \item For all $S\in \curg$, there is a block
            $B_S$ of $\lceil(1-\delta) |S|\rceil$ identical nodes in $R$.
          \item For all $S\in \curg$ and $u\in S$, there is an edge
            between $u$ and all nodes in $B_S$.
          \end{compactenum}
        \item Find a matching that covers $R$.
        \item
          For all $S\in \curg$, $\widehat{S}\gets$ nodes matched to
          the block $B_{S}$.
          Return $(\widehat{S},...\mid S\in \curg)$.
        \end{compactenum}
      \end{minipage}
    }
  \end{floatrow}
  \vspace{0.5ex}
  \begin{floatrow}[1]
    \killfloatstyle\floatsetup[figure]{style=plain}
    \ffigbox[\Xhsize]
    {
      \caption{The graph constructed by \textsc{Unravel} on input
        $S_1=\{a,b,c,e\}$, $S_2=\{d,e,f,g\}$, $S_3=\{g,h,i,j,k\}$, $S_4=\{k\}$
        for $\delta=\frac14$.}
      \label{fig:unravel:ex}
    }
    {
      \begin{tabular}[h]{cc}
        \includegraphics[width=0.48\textwidth]{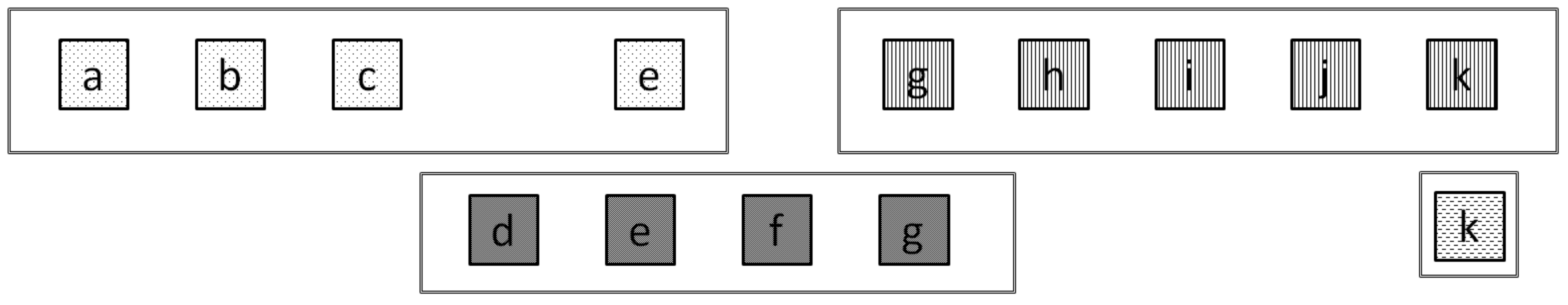}
        &
        \includegraphics[width=0.48\textwidth]{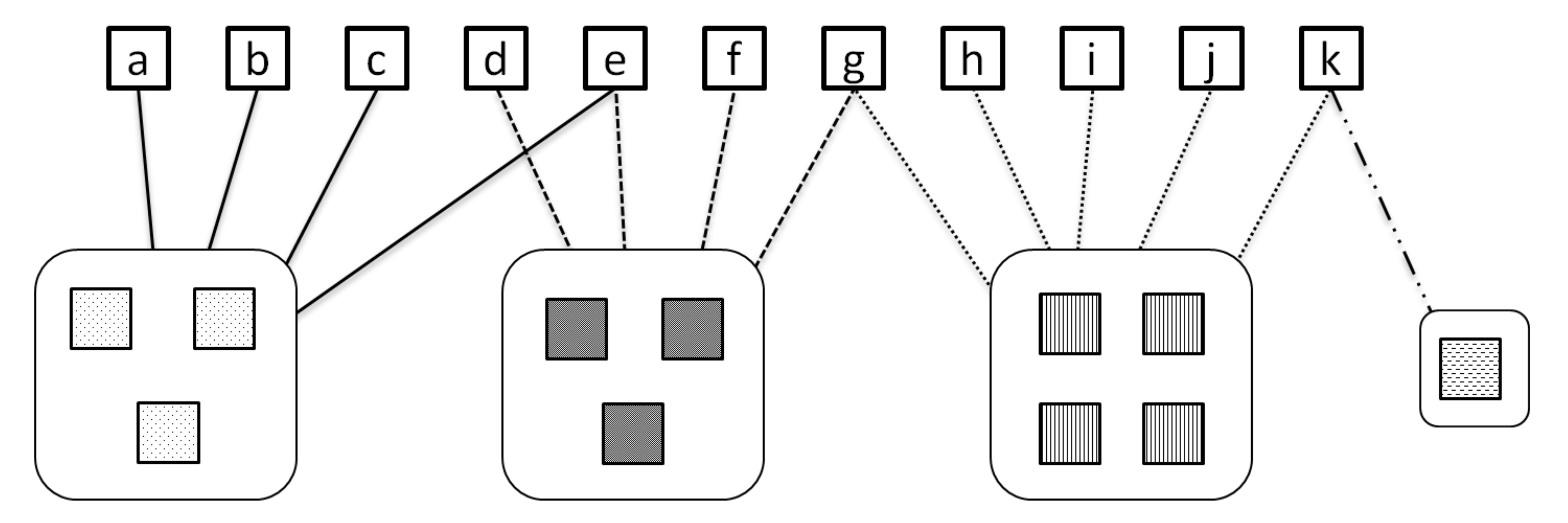}
        \\
        Input $\curg=\{S_1,S_2,S_3,S_4\}$.
        &
        Bipartite Graph.
      \end{tabular}
    }
  \end{floatrow}

\end{program*}

\setlength{\leftmargini}{2.5em}
\setlength{\leftmarginii}{2.2em}

\section{Analysis of the Algorithm}
\label{sec:analysis}
In this section, we prove the correctness of our algorithm. Our main
result is the following. Its proof is given at the end of
\Cref{sec:analysis:round}.
\begin{mytheorem}[Restatement of Theorem~\ref{thm:spectral-clustering}]
  Let $\optg \in \disjv(k)$ with $\|Y \moptg^\perp\|^2_2 \le O(\eps)$.
  Then $\newg \gets \textsc{SpectralClustering}(Y)$ is a $k$-partition
  so that $\newg \in \disjv(k)$, and it is $O(\sqrt{\eps})$ close to
  both $\optg$ and $Y$:
  \[
  \simfun(\optg, \newg) \le O( \sqrt{\eps} )
  \quad\mbox{and}\quad
  \| Y \mnewg^\perp \|^2_2 \le O(\sqrt{\eps}).
  \]
\end{mytheorem}
In order to keep the analysis simple, we make no effort toward
optimizing the constants.
\subsection{Preliminaries}
\label{sec:analysis:prelims}
In the following proposition, we show that for any pair of symmetric
matrices that are close to each other in spectral norm, if there is a
gap between the largest and second largest eigenvalues; then the
largest eigenvectors of both matrices will be very close to each
other. This can also be obtained using Wedin's theorem~\cite{ss90},
but we chose to give a simple and self-contained proof.
\begin{myproposition}
  \label{thm:eig-vec-dist}
  Given $A, B \in \symm^n$ with maximum eigenvectors
  $p, q \in \R^n$:
  \[
  \langle \unit{p} , \unit{q} \rangle^2 \ge 1 - \frac{2 \|A -
    B\|_2}{\sigma_1(A) - \sigma_2(A)}.
  \]
\end{myproposition}
\begin{proof} Suppose $\|p\|=\|q\|=1$.  Let
  $\delta \triangleq \|A-B\|_2$ and
  $\theta \triangleq \langle p , q \rangle^2$.  We have
  \begin{align*}    
    \sigma_1(A) = & p^T A p \le
    \delta + 
    \sigma_1(B) \langle q , p \rangle^2
      + \sigma_2(B) \|p^\perp q\|^2 \\
    \le& \delta + (\sigma_1(A) + \delta) \theta
     + (\sigma_2(A) + \delta)(1-\theta)\\
    =& \delta + \theta ( \sigma_1 - \sigma_2) + (\sigma_2 + 2 \delta).
  \end{align*}
  Hence $\theta \ge \frac{ \sigma_1 - \sigma_2 - 2 \delta}{\sigma_1 - \sigma_2}$.
\end{proof}
The main tool we use to identify the clusters will be the eigenvalues
of principal minors of the Gram matrix, $Y^T Y$. Basically,
eigenvalues measure how much true clusters, $T$, overlap with given
principal minor. We make this connection formal in the following
claim:
\begin{myclaim}
  \label{thm:minor-s-props}
  Given $Y \in \R^{m\times n}$, $\curg \in \disjv(r)$ and subset
  $S\subseteq V$, let
  $\rho=\big(|S\cap T|/|T|\ \big|\ T\in \curg\big)$.  Then for any $i$:
  \[
  | \sigma_i^2(Y_S) - \dec{(\rho)}_i | \le \|Y^T Y - \mcurg^\proj\|_2.
  \]
  Here $\dec{(\rho)}_i$ is the $i^{th}$ largest element of $\rho$.
\end{myclaim}
\begin{proof}
  We have
  $\|Y^T Y - \moptg^\proj\|_2 \ge \|Y_S^T Y_S -
  (\moptg^\proj)_{S,S}\|_2$.
  Observe that the eigenvectors of $(\moptg^\proj)_{S,S}$ are
  $\unit{\ind{T'\cap S}}$ with corresponding eigenvalue
  $\frac{|T'\cap S|}{|T'|}$ over all $T' \in \optg$. Thus $\rho$'s are
  the eigenvalues of $(\moptg^\proj)_{S,S}$.
\end{proof}
Consider a principal minor corresponding to some $S$ whose largest
eigenvalue is large, and second largest eigenvalue is small. The
previous claim implies that there is a unique optimal cluster $T$
which is almost contained by $S$. However this is still not
sufficient: $S$ might be much larger than $T$. In the next lemma, we
show that, one can take the top right singular vector of $Y_S$ and
round (threshold) it to obtain another subset $\widehat{S}\subseteq S$
which is now very close to $T$.
\begin{mylemma}[Initial Guess]
  \label{thm:boost-initial}  
  Given $Y \in \R^{m\times n}$, $\optg \in \disjv(r)$
  with $\|Y^T Y - \moptg^\proj\|_2 \le \delta$, and subset
  $S \subseteq V$,
  let $q \in \R^S$ be the top right singular vector of $Y_S$.
  If we define $\sigma_1 \triangleq \sigma_1^2(Y_S)$ and
  $\sigma_2 \triangleq \sigma_2^2(Y_S)$,
  then the subset $\widehat{S}\subseteq S$ obtained by:
  \[\widehat{S}\gets \textsc{Round}(q)\] satisfies the following.
  For $T$ being $\argmax_{T\in \optg} \frac{|S\cap T|}{|T|}$:
  \[
    \frac{|T\cap \widehat{S}|}{\sqrt{|T| |\widehat{S}|}}
    \ge \sqrt{\sigma_1 - \delta} \Big(1 - \frac{ 4 \delta } { \sigma_1
      - \sigma_2 }\Big)
  \]
  and
  \[
    \forall T'\neq T\in \optg:\ 
    |T' \cap \widehat{S}| \le (\sigma_2 + \delta) |T'|.
  \]
\end{mylemma}
\begin{proof} %
  We will use $\sigma_1 \triangleq \sigma_1^2(Y_S) $,
  $\sigma_2 = \sigma_2^2(Y_S)$ and
  $\Delta \triangleq \sigma_1 - \sigma_2$.  Since $q$ is the top right
  singular vector of $Y_S$, $\|Y_S q\|^2 = \sigma_1$ and $\|q\|=1$.

  By Claim~\ref{thm:minor-s-props},
  $ \big| \frac{|T\cap S|}{|T|} - \sigma_1 \big|
  \le \delta $, which implies
  $\frac{|S\cap T|}{|T|} \ge \sigma_1 - \delta$.
  For any $T' \neq T$,
  $\frac{|S\cap T'|}{|T'|} \le \sigma_2 + \delta$.  
  Via Proposition~\ref{thm:eig-vec-dist}, we see that:
  \[
  \langle q , \unit{\ind{T \cap S}}\rangle^2 \ge  1 - \frac{2
    \delta}{\sigma_1 - \sigma_2} 
  \ge 1 - \frac{2 \delta}{\Delta}
  \implies \langle q , \unit{\ind{\widehat{S}}}\rangle^2
    \ge 1 - \frac{2 \delta}{\Delta}.
  \]
  Provided that $\delta \le \frac14 \Delta$, both
  $\langle q , \unit{\ind{T \cap S}}\rangle$ and
  $ \langle q , \unit{\ind{\widehat{S}}}\rangle$ have the same
  sign. Therefore:
  \[
  \langle \unit{\ind{T \cap S}} , \unit{\ind{\widehat{S}}}\rangle
  \ge
    \langle q , \unit{\ind{T \cap S}}\rangle
    \langle q , \unit{\ind{\widehat{S}}}\rangle
    - \| q^\perp \unit{\ind{T \cap S}}\|
      \| q^\perp \unit{\ind{\widehat{S}}}\|
    \ge 1 - \frac{4 \delta}{\Delta}.
  \]
  Consequently, using the fact $\widehat{S} \subseteq S$, we see that
  any $T'\neq T$ has $|T'\cap \widehat{S}| \le (\alpha+\delta) |T'|$
  and $T\cap S \cap \widehat{S} = T\cap \widehat{S}$:
  \[
    1 - \frac{4 \delta}{\Delta}
    \le
     \frac{ |T\cap \widehat{S}| }{\sqrt{|T\cap S| |\widehat{S}| }}
    \le
    \frac{1}{\sqrt{\sigma_1-\delta}}
    \frac{ |T\cap \widehat{S}| }{\sqrt{|T| |\widehat{S}| }}.
%    \tag*{\qedhere}
  \]
  \iffalse
  \[
    \frac{2 \delta}{\Delta} \ge 
    \Delta(q, {\ind{T \cap S}})
    \implies \Delta(\widehat{S}, T) \le
    3 \Big( \Delta(\ind{\widehat{S}}, q )
      +  \Delta(q, \ind{S\cap T})
      +  \Delta({S\cap T}, T) \Big).
  \]
  \[
    \Delta(\widehat{S}, T) 
    \le \frac{15 \delta}{\Delta} + 3 \alpha.
  \]
  \fi
\end{proof}
After we found new clusters, we iterate by projecting $Y$ onto the
orthogonal complement of their center. In the following lemma, we
prove that, as long as the clusters were close to optimal ones, the
projection preserves remaining clusters.
\begin{mylemma}
  \label{thm:progress}
  Given $Y: Y^T \in \orthk$ and $\curg \in \disjv(r)$
  with the cluster centers in $\curg$  being linearly independent,
  suppose there exists $\optg \in\disjv(k)$ of the form
  $\optg = \optg' \uplus \optg''$:
  $\optg' \in \disjv(r), \optg'' \in \disjv(k-r)$ such that:
  \begin{itemize}
  \item $\|\mcurg - \moptg'\|_2^2\le \alpha$.
  \item $\|Y \moptg^\perp\|^2_2 \le \eps$.
  \end{itemize}
  Then $\optg''$ is a good spectral clustering for
  $Z \triangleq (Y \mcurg)^\perp Y$, in the sense that
  $ \| Z (\moptg'')^\perp \|^2_2 \le \eps + \alpha.  $ In addition,
  $\sigma_1(Z) = \ldots = \sigma_{k-r}(Z)=1$, $\sigma_{k-r+1}(Z) = 0$.
\end{mylemma}
\begin{proof}
  Note that $Z$ has all singular values either $0$ or $1$:
  \[
     Z Z^T = (Y \mcurg)^\perp Y Y^T (Y \mcurg)^\perp
      = (Y \mcurg)^\perp.
  \]
  Moreover $(Y \mcurg)^\perp$ has rank
  $\rank(Y) - \rank(Y \mcurg) = k - r$.
  Since spectral clustering is invariant under change of basis, we can
  assume $Z \in \R^{(k-r)\times n}$ so that all singular values of $Z$
  are $1$.  This means $Z$ is orthonormal. It is obtained from $Y$ by
  a linear transformation, therefore $\optg$ is a good spectral
  clustering for $Z$ also:
  \begin{align*}
    \eps \ge & \|Y \moptg^\perp Y^T\|_2  \ge \|Z \moptg^\perp Z^T\|_2 \\
    = & \| Z Z^T - (Z \moptg) (Z \moptg)^T\|_2 \\
    = & \| I_{k-r} - (Z \moptg) (Z \moptg)^T\|_2.    
  \end{align*}
  In other words,
  \begin{align*}    
    (1-\eps) I_{k-r} \preceq &  Z \moptg^{\proj} Z^T
      = Z \big[(\moptg')^{\proj} + (\moptg'')^{\proj}\big]  Z^T \\
    \implies 
    \eps I_{k-r} \succeq & Z (\moptg'')^\perp Z^T - Z^T \moptg^{\proj} Z^T.
  \end{align*}
  Now we will upper bound
  $\|Z (\moptg')^{\proj} Z^T\|_2 = \|Z \moptg'\|_2 = \|(Y
  \mcurg)^\perp (Y \moptg') \|_2$ by a simple Cauchy-Schwarz:
  \begin{align*}     
    \| (Y \mcurg)^\perp (Y \moptg') \|_2^2
    = & \| (Y \mcurg)^\perp (Y \moptg') - (Y \mcurg)^\perp (Y \mcurg) \|_2^2 \\
    = & \| (Y \mcurg)^\perp Y (\moptg' - \mcurg) \|_2^2 \\
    \le & \| (Y \mcurg)^\perp Y\|_2^2 \|\moptg' - \mcurg\|_2^2 \\
    \le & \|Y\|_2^2 \|\moptg' - \mcurg\|_2^2 \le \|\moptg' - \mcurg\|_2^2
          \le {\alpha}.
  \end{align*}
  As a consequence,
  $
  \|Z (\moptg'')^\perp \|^2_2
  \le \eps + \| Z \moptg' \|^2_2 \le \eps + \alpha.
  $
\end{proof}
Now we will start with the proof of correctness for our unraveling
procedure.
\subsection{Correctness of~\textsc{Unravel}~(\Cref{alg:unravel})}
\label{sec:analysis:unravel}
We will now prove that if the input of
$\textsc{Unravel}(\curg)$~(\Cref{alg:unravel}) is a list of (possibly
overlapping) sets which are close to some $k$-partition (ground
truth), then the output will be a list of $k$ {\bf disjoint} sets
which are also close to the ground truth. Our algorithm is based on
formulating this as a simple maximum bipartite matching problem. 
\begin{mylemma}
  \label{thm:unravel-correct}
  Given $\curg \in \setv(k)$, if there exists $\optg \in \disjv(k)$
  which is $\delta$-close to $\curg$, then
  $\textsc{Unravel}(\curg)$~(\Cref{alg:unravel}) will output
  $\newg \in \disjv(k)$ such that:
  \begin{itemize}
  \item For each $S \in \curg$, there exists $U \in \newg$ with
    $U \subseteq S$ and $|U| \ge (1-\delta) |S|$.
  \item $\newg$ is $4 \delta$-close to $\optg$.
  \end{itemize}
\end{mylemma}

\begin{proof}  
  It is easy to see that if all blocks are matched, then
  the resulting assignment is a collection of $k$-disjoint subsets so 
  it has the first property. For the second property, consider any
  $S \in \curg$ with corresponding subsets $U \in \newg$ and
  $T \in \optg$.
  \[
  \simfun(U,S) =
  1 - \frac{|S\cap U|^2}{|S| |U|}
  = 1 - \frac{|U|}{|S|} = \frac{|S\setminus U|}{|S|} \le \delta.
  \]
  Hence
  \(
  \sqrt{\simfun(U,T)}
  \le
  \sqrt{\simfun(U,S)}+\sqrt{\simfun(S,T)}
  \le 2 \sqrt{\delta}.
  \)
  This implies $\newg$ is $4 \delta$-close to $\optg$.

  Now we will prove that if $\delta$-close $\optg \in \disjv(k)$
  exists, then there is always a matching of all blocks.  Suppose
  $\pi:\curg\leftrightarrow \optg$ is a matching which minimizes
  $\max_S \simfun(S, \pi(S))$.  
  Then:
  \[
  {  | S \cap \pi(S) |^2 } \ge (1- \delta) { |S| |\pi(S)|}
  \implies
  |S\cap \pi(S)| \ge (1-\delta) |S|.
  \]
  By Hall's theorem, we have to show that for any set of right nodes,
  $B \subseteq R$, $B$'s neighbors on the left, $N(B)$, are more than
  $B$: $|N(B)| \ge |B|$.
  Observe that if $B$ contains some nodes of block $B_S$, then adding
  the whole block $B_S$ to $B$ does not increase $|N(B)|$, because all
  nodes in $B_S$ have the same set of neighbors. Hence the only
  subsets we need to consider are of the form $B = \cup_{S\in A} B_S$
  over all $A \subseteq \curg$:
  \begin{align*}
    |N(B)| = & \big| \cup_{S\in A} S \big| \ge
    \big| \cup_{S\in A} (S\cap \pi(S)) \big| \\
    = &
    \sum_{S \in A} |S\cap \pi(S)|
    \ge (1-\delta) \sum_S |S|
    = |B|.
  % \tag*{\qedhere}    
  \end{align*}
\end{proof}

\subsection{Correctness of \textsc{FindCluster}~(\Cref{alg:find-cluster})}
\label{sec:analysis:find-cluster}
Here we prove that, given an orthonormal basis $Y$, if it is close to
a $k$-partition up to rotation, then
$\textsc{FindCluster}(Y)$~(\Cref{alg:find-cluster}) will output a set
which is close to one of the clusters in this $k$-partition.

\begin{mylemma} %
  \label{thm:find-cluster-correct}
  Given $Y \in \R^{m\times n}$, whose singular values are $0$ or $1$,
  if there exists $\optg \in \disjv(k)$ with
  $\|Y^T Y- \moptg^\proj\|^2_2 \le \delta$ for some small enough
  constant $\delta$, then
  $\textsc{FindCluster}(Y)$~(\Cref{alg:find-cluster}) will output
  ${S'} \subseteq V$ such that: %
  \begin{itemize}
  \item There exists $T \in \optg$ with
    $\simfun({S'}, T) \le O(\sqrt{\delta})$,
  \item Any $T' \neq T: T'\in\optg$ has
    $|T' \cap {S'}| \le O(\sqrt{\delta}) |T'|$.
  \end{itemize}
\end{mylemma}
\begin{proof}
  $\|Y^T Y - \moptg^\proj\|^2_2 \le \delta$ implies 
  $
  \sigma_{k}(Y) = 1
  $
  and
  $
  \sigma_{k+1}(Y) = 0.
  $
  Therefore 
  $
  \| Y \moptg^\perp \|^2_F \le \delta k.
  $
  In other words,
  \[
  \sum_{T \in \optg} (\|Y_T\|^2_F - \|Y \unit{\ind{T}}\|^2_2)
  =
  \sum_{T \in \optg} (1 - \|Y \unit{\ind{T}}\|^2_2)
  \le \delta k.
  \]
  So:
  \[
  \sum_{T \in \optg} \Big[\max(1,\|Y_T\|^2_F) - \|Y \unit{\ind{T}}\|^2_2\Big]
  \le 2 \delta k.
  \]
  As a consequence, there exists $T \in \optg$ and some $c\in T$
  such that:
  \[
  \|Y_T\|^2_F \ge \|Y \unit{\ind{T}}\|^2_2 \ge  1 - 2 \delta.
  \]
  \[
  4 \delta \ge
  \frac{1}{|T|} 
  \sum_{u\in T, v\in T} \| Y_u - Y_v \|^2_2
  \ge \sum_{u\in T} \|Y_u - Y_c\|^2.
  \]
   
  \[
  \frac{4 \delta}{1-2\delta} \ge \expct{u\sim \|Y_u\|^2} {
    \frac{\|Y_u - Y_c\|^2}{\|Y_u\|^2}
  }.
  \]
  Let's define $\rho_u \triangleq \frac{\|Y_u - Y_c\|^2}{\|Y_u\|^2}$
  and sort the nodes in ascending order so that
  \(\rho_1 \le\rho_2 \le \ldots \le \rho_n:\)
  \[
  \expcts{u\sim \|Y_u\|^2}{\rho_u} \le \frac{4 \delta}{1- 2\delta} \le 9 \delta
  \qquad\mbox{provided $\delta \le \frac14$}.
  \] 
  By a simple Markov inequality, sum of all $\|Y_u\|^2$ over $u\in T$
  with $\rho_u \le 3 \sqrt{\delta}$ is at least $1 - 3 \sqrt{\delta}$.
  Consequently, the smallest integer $m$ for which
  $\sum_{1\le u\le m} \|Y_u\|^2 \ge 1 - 3 \sqrt{\delta}$ satisfies
  $\sum_{1\le u\le m} \|Y_u - Y_c\|^2 \le 3 \sqrt{\delta}$.

  From now on, we assume $S$ is a subset and
  $\delta': \delta' \le 3 \sqrt{\delta}$ with:
  \[
  \|Y_S\|^2_F
  = \sum_j \sigma_j^2(Y_S)
  \ge 1-\delta'
  \quad\mbox{and}\quad
  \delta' \ge
  \sum_{u\in S} \|Y_u - Y_c\|^2.
  \]
  Recall that variance is lower bounded by sum of the squares of all
  but largest singular values:
  \[\sum_{u\in S} \|Y_u - Y_c\|^2 \ge 
  \|Y_S\|^2_F - \sigma_1^2(Y_S) =
  \sum_{j\ge 2} \sigma_j^2(Y_S). \] 
  Hence $\sigma_1^2(Y_S) \ge 1 - 2 \delta'$ and
  $\sigma_2^2(Y_S) \le \delta'$.
  Provided that $\delta \le \frac{1}{100}$, which implies
  $\delta' \le \frac{3}{10}$:
  \[
  (\sigma_1-\delta) \Big(1 - \frac{4 \delta}{\sigma_1 - \sigma_2}
  \Big) \ge
  1 - 8 \sqrt{\delta}.
  \]
  For such $S$, Lemma~\ref{thm:boost-initial} tells us that the new subset
  $\widehat{S} \subseteq S$ obtained by rounding the top right
  singular vector of $Y_S$ satisfies, for
  $T \triangleq \argmax_{T\in \optg} \frac{|S\cap T|}{|T|}$:
  \begin{itemize}
  \item
    $\simfun(\widehat{S}, T) \le 1 - (\sigma_1-\delta) \Big(1 - \frac{ 4 \delta
    }{\sigma_1 - \sigma_2} \Big)
    \le 8 \sqrt{\delta}
    $.
  \item For any $T' \neq T \in \optg$,
    $|T' \cap \widehat{S}| \le 4 \sqrt{\delta} |T'|$. %{\qedhere}
  \end{itemize}
\end{proof}

\subsection{Correctness of \textsc{Boost}~(\Cref{alg:boost})}
\label{sec:analysis:boost}
As we mentioned earlier, if we keep finding clusters and removing them
iteratively, the error will quickly accumulate and degrade the quality
of remaining clusters. To prevent this, we apply a boosting procedure
as described in
$\widehat{S}\gets \textsc{Boost}(Y,S)$~(\Cref{alg:boost}) every time
we find a new cluster. The main idea is that, if $S$ is close to some
cluster in ground truth, say $T$, and far from others; then the top
left singular vector, say $p$, of the vectors associated with $S$ will
be close to the ones $S\cap T$.  Unfortunately, we can not use simple
perturbation bounds such as Wedin's theorem. We have to make full use
of \cref{eq:main-goal-1} instead: Under projection by $p$, the vectors
of $T$ tend to stay together; therefore vectors in $T \setminus S$
will be very close to the vectors in $S\cap T$. Hence indeed $p$ will
be close most of $T\setminus S$ in addition to $S\cap T$.

% As we mentioned earlier, if we keep finding clusters and removing them
% iteratively, the error will quickly accumulate and degrade the quality
% of remaining clusters. To prevent this, we apply a boosting procedure
% as described in
% $\widehat{S}\gets \textsc{Boost}(Y,S)$~(\Cref{alg:boost}) every time
% we find a new cluster. The main idea is that, if $S$ is close to some
% cluster in ground truth, say $T$, and far from others; then the top
% left singular vector, say $p$, of the vectors associated with $S$ will
% be close to the ones $S\cap T$. On the other hand, under any
% projection, the vectors of $T$ tend to stay together; therefore
% vectors in $T \setminus S$ will be very close to the vectors in
% $S\cap T$. Therefore $p$ will also be close most of $T\setminus S$.
%

%
\begin{mytheorem}[Boosting]
  \label{thm:boost-main}  
  Given $Y \in \orthk$ and $\optg \in \disjv(k)$
  with $\|Y \moptg^\perp \|^2_2 \le \eps$, consider any subset
  $S \subseteq V$. Suppose there exists $T \in \optg$ with
  $|S\cap T| \ge (1 - \alpha) |T|$ such that for any
  $T'\neq T \in \optg$, $|S\cap T'| \le \alpha |T'|$
  for some $\alpha \le \frac14$.
  Then for $\widehat{S}\gets \textsc{Boost}(Y,S)$~(\Cref{alg:boost}),
  $\widehat{S}$
  satisfies:
  \[
    \simfun( \widehat{S} , T ) \le c_0 \sqrt{\eps}
  \] for some constant $c_0 \le 50$.
\end{mytheorem}
\begin{remark}
  Note that Theorem~\ref{thm:boost-main} allows us to convert a subset with
  non-negligible overlap into a subset which is very
  close. Unfortunately, the new subset we obtain is no longer
  guaranteed to have small intersection with other $T' \neq T$.
\end{remark}
\begin{proof}[of Theorem~\ref{thm:boost-main}]
  By Lemma~\ref{thm:sandwich}, for $\delta \triangleq  \sqrt{\eps}$:
  \[
    \|Y^T Y - \moptg^\proj\|_2 \le \delta.
  \]
  We will use $\sigma_1 \triangleq \sigma_1^2(Y_S) $ and
  $\sigma_2 = \sigma_2^2(Y_S)$. 
  By Claim~\ref{thm:minor-s-props},
  $\big|\sigma_1 - |S\cap T|/|T|\big| \le \delta$ and
  $\sigma_2 \le \beta+\delta$.

  Since $p$ is the left top singular vector of $Y_S$,
  $\|Y_S^T p\|^2 = \sigma_1$. We define $q\triangleq Y^T p$ so that
  $\|q_S\|^2 = \sigma_1$:
  \begin{align}
    \frac{\|q_S\|^2}{|S\cap T|/|T|}
    \ge & 1 - \frac{\delta}{|S\cap T|/|T|} \notag \\    
    \ge &  1 - \frac{\delta}{1-\alpha}
          \ge 1 - \frac{4 \delta}{3}
          \label{eq:pf:eig-vec-dist:1} .
  \end{align}
  Since
  $\sigma_1((\moptg^\proj)_{S,S})-\sigma_2((\moptg^\proj)_{S,S})\ge
  1-2\alpha$, Proposition~\ref{thm:eig-vec-dist} implies:
  \[
  \langle \unit{\ind{T \cap S}} , q \rangle^2
  \ge \|q_S\|^2 \Big( 1 - \frac{2 \delta}{1-2\alpha} \Big)
  \ge \|q_S\|^2 \Big( 1 - 4 \delta  \Big).
  \]
  Let's use $\mu_A$ to denote the mean of $p$ on subset $A$,
  $\mu_A \triangleq \langle \dunit{\ind{A}} , q \rangle$.  We will
  assume, without loss of generality, $\mu_{S\cap T} \ge 0$.
  Hence
  $\mu_{S\cap T} \ge \frac{\|q_S\|}{\sqrt{|S\cap T|}} \sqrt{1 - {4 \delta}}$.
  On the other hand, 
  \begin{align*}    
    \eps \ge & p^T Y \moptg^\perp Y^T p 
               = q^T \moptg^\perp q \\
    \ge & 
               \frac{1}{|T|} \sum_{\{i,j\}\in \binom{T}{2}} (q_i - q_j)^2
               %q^T K_T q
               \ge |S\cap T| (\mu_{S\cap T} - \mu_{T})^2. \\
     \mu_T \ge & \mu_{S\cap T} - \sqrt{\frac{\eps}{|S\cap T|}}.
                 \\
    \langle \unit{\ind{T}} , q \rangle
                 \ge & \sqrt{|T|} \mu_{S\cap T} - \sqrt{\frac{\eps
                       |T|}{|S\cap T|}}
                       \\
    \ge & \sqrt{\frac{|T|}{|S\cap T|}}
                       \Big(
                       \|q_S\| \sqrt{1 - {4 \delta}{}}
                       - \sqrt{\eps}
                       \Big).
\intertext{Using~\cref{eq:pf:eig-vec-dist:1}, %
           we can lower bound this quantity as:}
    \ge & \sqrt{ \Big(1 - \frac{4 \delta}{3} \Big)
          \Big( 1 - {4 \delta}{} \Big) }
          -  \sqrt{\frac{\eps}{1 - \alpha}} \\
    \ge & 1 - {4 \delta}{} 
          -  \sqrt{\frac{4 \eps}{3}}
          \ge 1 - \frac{26}{5} \delta. 
  \end{align*}
  Since $\|q\| = \|Y^T p\| \le 1$, we have
  \(
  \simfun(p, T)
  \le 1 - \langle \unit{\ind{T}} , q \rangle^2
  \le 11 \delta.
  \)
  Using Proposition~\ref{thm:round-correct}, we see that
  \(\widehat{S} \gets \textsc{Round}(q)\) satisfies:
  \[\simfun(\widehat{S}, T) \le 4 \simfun(q, T)
  \le 
    44 \delta = 44 \sqrt{\eps}.
%    \tag*{\qedhere}
  \]  
\end{proof}

\subsection{Correctness of~\textsc{Round}~(\Cref{alg:round-vec})}
\label{sec:analysis:round}
Possibly the simplest case for our problem is when $Y$ is
$1$-dimensional, i.e. it is a vector. As we argued in the
introduction, our algorithm boils down to simple thresholding in this
case. % as given in \textsc{Round}~(\Cref{alg:round-vec}).

\begin{myproposition}
  \label{thm:round-correct}
  Given \(q\neq 0 \in \R^n\),
  for any \(T\neq \es\),
  \(S \gets \textsc{Round}(q)\) (\Cref{alg:round-vec})
  satisfies \( \simfun(S, T) \le 4 \simfun(q, T).  \)
\end{myproposition}
\begin{proof}
  Let \(\eps\triangleq \simfun(q,T)\).
  Without loss of generality, we may assume
  \(q_1 \ge \ldots \ge q_m > 0 \ge \ldots \ge q_n\),
  \(\|q\|=1\) and \(\langle q , \ind{T} \rangle \ge 0\).
  We have:
  \
  \begin{align*}
    \sqrt{1-\eps} \le & \frac{\langle q , {\ind{T}} \rangle}{\sqrt{|T|}}
                        \le \frac{\sum_{j\le |T|} q_j}{\sqrt{|T|}} \\
    \le & \max_{S'} \frac{\sum_{j\le |S'|} q_j}{\sqrt{|S'|}}
    \le \frac{\sum_{j\le |S|} q_j}{\sqrt{|S|}}
    \le | \langle q , \unit{\ind{S}}\rangle |.    
  \end{align*}
  So \(\simfun(S,q) \le \eps\)
  and
  \( \sqrt{\simfun(S,T)} \le \sqrt{\simfun(S,q)}+\sqrt{\simfun(q,T)}
  \le 2 \sqrt{\eps} \). %
\end{proof}

\subsection{Correctness of
  \textsc{SpectralClustering}~(\Cref{alg:spectral-clustering})}
\label{sec:analysis:sc}
Finally, we put everything together and prove the correctness of
$\textsc{SpectralClustering}(Y)$ (\Cref{alg:spectral-clustering}).  In
the following lemma, we will show that the algorithm will iteratively
find sets $S_1,S_2,\ldots$ (think of them as coarse approximations of
$T_i$'s) and $\widehat{S}_1,\widehat{S}_2,\ldots$ (think of them as
fine approximations of $T_i$'s) such that at each iteration, each
$S_i$ and $\widehat{S}_i$ will correspond to a unique $T_i$. Moreover,
each $S_i$ will have very small overlap with remaining $T_j$'s coming
after themselves. Even though $S_i$'s might still have large overlap
with previous $T_j$'s for $j<i$, we can easily use
$\textsc{Unravel}(\curg)$~(\Cref{alg:unravel}) to rectify this issue.
\begin{mylemma}
  \label{thm:spectral-clustering-induct}
  Let $\optg \in \disjv(k)$ with $\|Y \moptg^\perp\|^2_2 \le \eps$ for
  some $\eps \le \eps_0$, where $\eps_0 \in (0,1)$ is a constant.  For
  any $r \in [k]$, consider the sequences $\curg = (S_1, \ldots, S_r)$
  and $\newg = (\widehat{S_1}, \ldots, \widehat{S_r})$ as found by
  $\textsc{SpectralClustering}(Y)$ at the start of $r^{th}$
  iteration. Then there exists an ordering of $\optg$:
  \[
  \optg=( \underbrace{ T_1, T_2,\ldots, T_r }_{\triangleq \optg'}  ,
  \underbrace{T_{r+1}, \ldots, T_k}_{\triangleq \optg''} ),
  \]
  with the following properties for some $\alpha \le \frac{1}{16}$ and
  $\beta \le 100$:

  \begin{enumerate}[(a)]
  \item \label{item:sc:1}
    For every $i \le  r$:
    \begin{itemize}
    \item $\simfun(S_i, T_i) \le \alpha$.
    \item For all $j > i$, $|T_j \cap S_i| \le \alpha |T_j|$.
    \end{itemize}
  \item \label{item:sc:2} 
    For every $i \le r$,
    $\simfun(\widehat{S_i}, T_i) \le \beta \sqrt{\eps}$.
  \end{enumerate}
\end{mylemma}
\begin{proof}
  By induction on $r$. For $r=0$, \eqref{item:sc:1} and
  \eqref{item:sc:2} are trivially true.

  Given $r$, suppose \eqref{item:sc:1} and \eqref{item:sc:2} %
  are true with $(T_1,\ldots, T_k)$, $\optg'$ and
  $\optg''$ being as described.

  At the beginning of $(r+1)^{st}$ iteration, we have
  $\curg = (S_1, \ldots, S_r)$ and
  $\newg = (\widehat{S_1}, \ldots, \widehat{S_r})$.  
  \eqref{item:sc:2} means $\newg$ is $\beta \sqrt{\eps}$-close to
  $\optg'$.

  After $\newg' \gets \textsc{Unravel}(\newg)$,
  by Lemma~\ref{thm:unravel-correct}, $\newg'$ is
  $4 \beta \sqrt{\eps}$-close to $\optg$ and $\newg' \in \disjv(r)$.
  Using Theorem~\ref{thm:set-similarity}, we see that
  $\|\mnewg' - \moptg'\|^2_2 \le 8 \beta \sqrt{\eps}$. 
  Now we can invoke Lemma~\ref{thm:progress}, which implies:
  \[
  \| Z (\moptg'')^\perp \|^2_2 \le \eps + 8 \beta \sqrt{\eps}
  \le 9 \beta \sqrt{\eps}.
  \] 
  Provided 
  $ 9 \beta \sqrt{\eps} \le \delta_0 $ for some small enough constant
  $\delta_0$,
  we can use Lemma~\ref{thm:find-cluster-correct} to see that the subset
  $S_{r+1} \gets \textsc{FindCluster}( Z )$ satisfies 
  \[
  \simfun(S_{r+1} , T) \le \alpha, %
  \]
  and
  \[
  \forall T' \in \optg'': T' \neq T \implies |T'\cap S_{r+1}| \le \alpha |T'|.
  \]
  We reorder $(T_{r+1},\ldots, T_k)$ so that $T_{r+1} = T$. Then
  \eqref{item:sc:1} holds true when $i=r+1$.  Since $T_1,\ldots, T_r$
  remain the same, \eqref{item:sc:1} is true for all $i\le r+1$.

  Consider $(U_1,\ldots, U_{r+1}) \gets \textsc{Unravel}(\curg)$.
  By Lemma~\ref{thm:unravel-correct}, we know that
  $\simfun(U_i, T_i) \le 4 \alpha$ and
  $U_i \subseteq S_i, |U_i| \ge (1-\alpha) |S_i|$ for each
  $i \in [r+1]$. In particular, for all $i\le r+1$:
  \[ 
  |U_i \cap T_i| \ge (1-4\alpha) |T_i|. 
  \]
  $U_{r+1}$ being a subset of $S_{r+1}$ means
  $|U_{r+1} \cap T_j| \le |S_{r+1} \cap T_j| \le \alpha |T_j|$
  whenever $j > r+1$. Now we will prove the case of $j \le r$.  Using
  the fact that $U$'s are disjoint, for any $j \le r$:
  \[
  |U_{r+1} \cap T_j| \le |T_j| - |T_j \cap U_j|
  \le |T_j| - (1-4 \alpha) |T_j| = 4 \alpha |T_j|.
  \]
  Consequently, for any $j \neq r+1$:
  \[
  |U_{r+1} \cap T_j | \le 4 \alpha |T_j|. 
  \]
  After executing
  $\widehat{S_{r+1}} \gets \textsc{Boost}(Y, U_{r+1})$, noting
  $\alpha \le \frac{1}{16}$, we see via Theorem~\ref{thm:boost-main}:
  \[
  \simfun(\widehat{S_{r+1}}, T_{r+1}) \le c_0 \sqrt{\eps} = \beta \sqrt{\eps}.
  \]
  Combined with the fact that $\widehat{S_{i}}$ and $T_i$ remain the
  same for $i\le r$, \eqref{item:sc:2} also remains true for all
  $i\le r+1$.
  By induction, we now see that both~\eqref{item:sc:1} and
  \eqref{item:sc:2} are true for all $r\le k$.
\end{proof}

\begin{mytheorem}
  \label{thm:spectral-clustering}
  Let $\optg \in \disjv(k)$ with $\|Y \moptg^\perp\|^2_2 \le O(\eps)$.
  Then $\newg \gets \textsc{SpectralClustering}(Y)$ is a $k$-partition
  so that $\newg \in \disjv(k)$, and it is $O(\sqrt{\eps})$ close to
  both $\optg$ and $Y$:
  \[
  \simfun(\optg, \newg) \le O( \sqrt{\eps} )
  \quad\mbox{and}\quad
  \| Y \mnewg^\perp \|^2_2 \le O(\sqrt{\eps}).
  \]
\end{mytheorem}
\begin{proof}
  By Lemma~\ref{thm:spectral-clustering-induct},
  $\newg = (\widehat{S_1},\ldots, \widehat{S_{k}})$ is
  $\beta \sqrt{\eps}$-close to $\optg$. Lemma~\ref{thm:unravel-correct}
  implies that $\textsc{Unravel}(\newg)$ outputs a disjoint collection
  of $k$-subsets which is $4 \beta \sqrt{\eps}$-close to $\optg$. For
  the second bound:
  \begin{align*}    
    \frac12 \| Y \mnewg^\perp \|^2_2
    \le& \| Y \moptg^\perp \mnewg^\perp \|^2_2
         +  \| Y \moptg^\proj \mnewg^\perp \|^2_2 \\
    \le & \| Y \moptg^\perp\|^2_2 \cdot \| \mnewg^\perp \|^2_2
          +  \| Y \moptg \|^2_2 \cdot \|\moptg^T \mnewg^\perp \|^2_2 \\
    \le & \eps +  \|\moptg^T \mnewg^\perp \|^2_2 
    \le \eps +  2 \simfun(\optg, \newg) \le O(\sqrt{\eps}).
  \end{align*}
  In the second to last inequality, we used Theorem~\ref{thm:set-similarity}.
\end{proof}
\section{Applications}
\label{sec:applications}
In this section, we will show some applications of our spectral
clustering algorithm.

\subsection{\minmaxkpar}
\label{sec:app-minmakpar}
Our first application is approximating non-expanding $k$-partitions in
graphs. One may also interpret this as applying our subspace rounding
algorithm on the {\em basic} SDP relaxation for \minmaxkpar\ problem.
\begin{mytheorem}
  \label{thm:app-min-max-exp}
  Given a graph $G$ with Laplacian matrix $L$, let $\curg$ be the
  $k$-partition obtained by running~\Cref{alg:spectral-clustering} on
  the smallest $k$ eigenvectors of $L$.
  Then:
  \[
  \Delta(\curg, \optg) \le O\Bigg(\sqrt{\frac{\phi_k}{\lambda_{k+1}}}\Bigg).
  \]
\end{mytheorem}
\begin{remark}[Faster Algorithm]
  By slightly modifying our algorithm to take advantage of the
  underlying graph structure, one can obtain a faster randomized
  algorithm having the same guarantees with Theorem~\ref{thm:app-min-max-exp}
  with expected running time $O(k^2 (n+m))$.
\end{remark}
\begin{proof}
  From Lemma~\ref{thm:minmaxkpar-as-snorm}, we know that
  \(
  \phi_k \le 
  \sigma_{\max}(\moptg^T L \moptg) %
  \le 2 \phi_k.
  \)
  Now consider the matrix $Z = Y^T$, whose columns are the smallest
  $k$ eigenvectors of $L$. We have
  \( L \succeq \lambda_{k+1} Z^\perp \) which means:
  \[
  \lambda_{k+1}\cdot \moptg^T Z^\perp \moptg \preceq 
  \moptg^T L \moptg \preceq 2 \phi I_k.
  \]
  This implies
  \[
  O(\eps) \ge 
  \|\moptg^T Z^\perp \moptg\|_2
  = \|I_k - \moptg^T Z Z^T \moptg\|_2.
  \]
  Thus
  $\sigma_k(Z^T \moptg) = \sigma_{\min}(Z^T \moptg)
  \ge
  \sqrt{1-O(\eps)}$ and:
  \begin{align*}
    \|Y \moptg^\perp\|^2_2 = & 
                               \|Z^T \moptg^\perp Z\|_2
                               = \| I_k - Z^T \moptg \moptg^T Z \|_2 \\
    = & 1 - \sigma_{\min}(\moptg^T Z)
        \le O(\eps).    
  \end{align*}
  The claim follows from Theorem~\ref{thm:spectral-clustering}.
\end{proof}

\subsection{Matrix and Graph Approximations}
\label{sec:app-apprx}
Our next application is for approximating a matrix in terms of
$k$-block diagonal matrices corresponding to the adjacency matrices of
normalized cliques, under spectral norm.
\begin{mytheorem}
  \label{thm:app-mtx-par}
  Given a matrix $X \in \symm^n$, let
  \(\eps \triangleq \min_{\optg \in \disjv(k)} \| X - \moptg^\proj
  \|_2\).
  In polynomial time, we can find $\curg\in\disjv(k)$ such that
  $\simfun(\optg,\curg) \le O(\sqrt{\eps})$ and:
  \[
  \| X - \mcurg^\proj\|_2 \le O\big(\eps^{1/4}\big).
  \]
\end{mytheorem}
\begin{proof}
  Let $Y$ be the matrix whose rows are the top $k$ eigenvectors of
  $X$. \linebreak{}
  Consider $\curg \gets \textsc{SpectralClustering}(Y)$:
  \[
  \| Y^T Y - X \|_2 \le \|X - \moptg \moptg^T \|_2 \le \eps
  \] which means
  \[
  \| Y^T Y - \moptg \moptg^T \|_2 \le 2 \eps
  \implies
  \| Y \moptg^\perp\|_2^2 \le 2 \eps.
  \]
  By Theorem~\ref{thm:spectral-clustering},
  $\simfun(\optg,\curg) \le O(\sqrt{\eps})$ and
  $\|Y \mcurg^\perp\|^2_2 \le O(\sqrt{\eps})$:
  \[
  \| \mcurg \mcurg^T - X \|_2 \le
  \eps + 
  \| \moptg \moptg^T - \mcurg \mcurg^T \|_2 \le
  O(\eps^{1/4}). %\tag*{\qedhere}
  \]
\end{proof}
Our final application is for approximating a graph Laplacian via
another Laplacian corresponding to the graph formed as a disjoint
union of $k$ normalized cliques (expanders), again under spectral
norm. Since we are working with Laplacian matrices, this means the new
graph approximates cuts of the original graph also.
\begin{mycorollary}
  \label{thm:app-graph-par}
  Given a graph $G$, if there exists $\optg \in \disjv(k)$ such that
  Laplacian of $G$ is $\eps$-close (in spectral norm) to the Laplacian
  corresponding to the disjoint union of normalized cliques on each
  $T \in \optg$:
  \[
    \| L - \moptg^\perp \|_2 \le \eps, 
  \]
  then we can find $\curg\in\disjv(k)$ which is
  $O\big(\sqrt{\eps}\big)$-close to $\optg$ and $G$ in polynomial time:
  \[
    \big\| L - \mcurg^\perp \|_2 \le O({\eps}^{1/4}).
  \]
\end{mycorollary}
\begin{proof}
  Since $\|L - \moptg^\perp\|_2 = \| (I-L) - \moptg^\proj \|_2$, we
  can apply Theorem~\ref{thm:app-mtx-par} on the matrix $I-L$. The
  rest follows easily.
\end{proof}
\section{\minmaxkpar\ Implies Spectral Clustering}
\label{sec:reduction-ncuts}
In this section, we will show that approximation algorithms for
various graph partitioning problems imply similar approximation
guarantees for our clustering problem.
\begin{mytheorem}
  \label{thm:y-to-x}
  Given $Y: Y^T\in \orthk$, let
  $\optg \triangleq \argmin_{\curg} \|Y \mcurg^\perp\|_2^2$ with
  $\eps = \|Y \moptg^\perp\|^2_2$. Then there exists a weighted,
  undirected, regular graph $X$, whose normalized Laplacian matrix has
  its $(k+1)^{st}$ smallest eigenvalue
  \(\lambda_{k+1}\) is at least
  \(
    \lambda_{k+1} \ge 1 - O(\sqrt{\eps})
  \)
  such that:
  \begin{itemize}
  \item Each $T\in \optg$ has small expansion, $\phi_X(T) \le O(\eps)$,
  \item If $\curg \in \disjv(k)$ is a $k$-partition with
    $\max_{S\in \curg} \phi_X(S) \le \delta$, then
    $
      \simfun(\curg, \optg) \le O(\delta + \sqrt{\eps}).
    $
  \end{itemize}
  Moreover such $X$ can be constructed in polynomial time.  
\end{mytheorem}
\begin{proof}
  Consider the following SDP. Here we chose $\eps\in [0,1]$ to be the
  minimum value where this SDP remains feasible:
  \begin{enumerate}[(i)]
  \item\label{itm:y-to-x:01} $X \preceq Y^\proj + \sqrt{\eps} Y^\perp$.
  \item\label{itm:y-to-x:0} $Y X Y^T \succeq (1-\eps) I_k$,
  \item\label{itm:y-to-x:1} \label{itm:y-to-x:2} \label{itm:y-to-x:3}
    \label{itm:y-to-x:4} 
    $X$ is doubly stochastic, diagonally dominant, PSD and has trace
    $k$.
  \end{enumerate}
  It is easy to see that $\moptg^\proj$ is a feasible solution
  (Lemma~\ref{thm:sandwich}). Moreover any feasible solution $X$
  corresponds to the adjacency matrix of a graph which is undirected
  and has all degrees equal to $1$. Now we will show that it has the
  other properties:
  \begin{itemize}
  \item \(
  \lambda_{k+1} = 1 - \sigma_{k+1}(X)
  \ge 1 -  \sigma_{k+1}(Y^\proj + \sqrt{\eps} Y^\perp)
  = 1 - \sqrt{\eps}.
  \) 
  \item \(\max_T \phi_X(T) \le 2 \|\moptg^T (I-X) \moptg\|_2
    \le O(\|Y^T (I-X) Y\|_2) + \eps \le O(\eps).\)
  \item Recall
    \(\max_S \phi_X(S)\ge \frac12 \| \mcurg^T (I-X) \mcurg \|_2 \)
    so \(\sigma_k(Y^T \mcurg) \ge 1 - O(\delta + \eps).\)
    In particular
    \(\sigma_k(\moptg^T \mcurg) \ge 1 - O(\sqrt{\eps} + \delta).\)
    Using Theorem~\ref{thm:set-similarity}, we have
    \(\simfun(\curg,\optg) \le O(\sqrt{\eps} + \delta).\) %\qedhere
  \end{itemize}
\end{proof}

\section{Omitted Proofs}
\subsection{Proof of Theorem~\ref{thm:set-similarity}}
\label{sec:proof-of-set-similarity}
\begin{mytheorem}[Restatement of Theorem~\ref{thm:set-similarity}]
  Given $\curg, \newg\in \disjv(k)$;
  $
    \simfun(\curg,\newg) \le 
    \| \mcurg^\perp \mnewg \|^2_2 \le 2 \simfun(\curg,\newg).
  $
  Moreover, after appropriately ordering the columns of
  $\mnewg$, $\|\mcurg - \mnewg\|_2^2 \le 4 \simfun(\curg,\newg)$.
\end{mytheorem}
\begin{proof}[Upper Bound]
  Recall that $\mcurg^\perp$ is a Laplacian matrix.
  From Lemma~\ref{thm:minmaxkpar-as-snorm}, we see that
  \(
    \|\mnewg^T \mcurg^\perp \mnewg\|_2
  \) is 
  within factor-$2$ of the maximum diagonal element. Given
  $T \in \newg$ let $S$ be the matching set in $\curg$ so that
  $\simfun(S,T) \le \simfun(\curg,\newg)$.
  The diagonal element of \(\mnewg^T \mcurg^\perp \mnewg\)
  corresponding to $T$ is:
  \[
     \unit{\ind{T}}^T \mcurg^\perp \unit{\ind{T}}
      \le {\|\ind{S}^\perp \unit{\ind{T}}\|^2}      
      = 1 - \frac{|S\cap T|^2}{|S| |T|}
      = \simfun(S,T) \le \eps.
  \]
  Therefore $\|\mcurg^\perp \mnewg\|^2_2 \le 2 \eps$.  Before proving
  the lower bound, we will show how this bounds
  $\|\mcurg - \mnewg\|_2$:
  \[
    \frac12 (\mcurg - \mnewg)^T (\mcurg - \mnewg)    
    = I_k - \frac{1}{2} (\mcurg^T \mnewg + \mnewg^T \mcurg).
  \]
  So $\frac12 \|\mcurg - \mnewg\|^2_2 \le 1 - \sigma_{\min}(\mnewg^T \mcurg)
  = 1 - \sqrt{1 - \|\mcurg^\perp \mnewg\|^2_2} \le 2 \eps$.
\end{proof}
\def\gga{{\curg}}
\def\ggb{{\newg}}
\def\ggap{{\mcurg}}
\def\ggbp{{\mnewg}}
\begin{proof}[Lower Bound]
  We define 
  $\pi_1: \gga \to \ggb$ and $\pi_2: \ggb \to \gga$
  as the following:
  \begin{align*}
  \forall S\in \gga:
  \pi_1(S) \triangleq & \argmax_{T\in \ggb} \frac{|S\cap T|}{|T|}, 
    \\
  \forall T\in \ggb:
  \pi_2(T) \triangleq & \argmax_{S\in \gga} \frac{|S\cap T|}{|S|}.    
  \end{align*}
  Consider $M = \{ (S, \pi_1(S))\ |\ S \in \gga\}$: By
  Claims~\ref{thm:pi-1-is-bijection} and \ref{thm:pis-match}, $M$ is
  indeed a perfect matching between $\gga$ and $\ggb$. Now consider
  any matched pair $(S,T) \in M$. Without loss of generality, say
  $|S| \ge |T|$.  By Claim~\ref{thm:pi-1-large},
  $|S \cap T| \ge (1-\eps) |S|$.  Since
  $|S\Delta T| = |S| + |T| - 2 |S\cap T|$:
  \[
  |S\Delta T| \le |S| + |T| - 2 (1-\eps) |S| =
  2\eps |S| + (|T|-|S|) \le 2 \eps |S|.
  \]
  We finish our proof with Claims~\ref{thm:pi-1-large},
  \ref{thm:pi-1-is-bijection} and \ref{thm:pis-match}.
  \begin{myclaim}\label{thm:pi-1-large}
    If $\pi_1(S) = T$, then $|S\cap T| \ge (1-\eps) |T|$.	
    Similarly, if $\pi_2(T) = S$, then $|S\cap T| \ge (1-\eps) |S|$.
  \end{myclaim}
  \begin{proof}
    Consider the matrix 
    $P = \ggap^T \ggbp \ggbp^T \ggap \in \sympm^k$
    so that $\lambda_{\min}(P) = \sigma_{\min}^2(\ggap^T \ggbp)$.
    Thus
    $\lambda_{\min}(P) = \sigma_{\min}(\ggap^T \ggbp)^2
    \ge 1-  \eps$.
    In particular, all diagonals of $P$ are at least $1- \eps$.
    Consider any diagonal corresponding to $S \in \gga$:
    \begin{align*}
    1-\eps\le&
               \frac{\ind{S}^T \ggbp \ggbp^T \ind{S} }{|S|}
               = \sum_{T\in \ggb} \frac{|S\cap T|^2}{|S| |T| } \\
      \le & 
            \Bigg(\max_{T'\in \ggb} \frac{|S\cap T'|}{|T'|}\Bigg)
            \sum_{T\in \ggb} \frac{|S\cap T|}{|S|} \\	
      = & \max_{T'\in \ggb} \frac{|S\cap T'|}{|T'|},      
    \end{align*}
    which, by construction, is equal to 
    $\frac{|S\cap \pi_1(S)|}{|\pi_1(S)|}$. This proves
    the first part of the claim. 
    The second part follows immediately by
    applying the same argument on $\ggb$ and $\gga$.
  \end{proof}
  \begin{myclaim} \label{thm:pi-1-is-bijection}
    Both $\pi_1$ and $\pi_2$ are bijections.
  \end{myclaim}
  \begin{proof}
    Suppose $\pi_1(S) = \pi_1(S') = T$ for some $S \neq S'$. Since 
    $S,S'$ are disjoint and $\eps<\frac12$:
    \[
    |T| \ge 
    |S\cap T| + |S'\cap T| \ge 2 (1-\eps) |T| 
    > |T|,
    \] a contradiction. A similar argument shows that $\pi_2$
    is a bijection as well.
  \end{proof}
  \begin{myclaim} \label{thm:pis-match}
    $\pi_1 = \pi_2^{-1}$. 
  \end{myclaim}
  \begin{proof}
    Suppose not. Since 
    both $\gga$ and $\ggb$ are bijections
    by Claim~\ref{thm:pi-1-is-bijection},
    there exists a cycle of the form 
    \[
    (S_0, T_0, \ldots, S_{m-1}, T_{m-1}, S_m=S_0)
    \] where 
    $\pi_1(S_i) = T_i$ and $\pi_2(T_{i}) = S_{i+1}$
    for some $m\ge 2$.
    By construction, 
    $ |S_i \cap T_i| \ge (1-\eps) |T_i| $ which means
    $\eps |T_i| \ge |T_i \sm S_i|$.  Since $S_i$ and $S_{i+1}$ are
    disjoint, $|T_i \sm S_i| \ge |T_i \cap S_{i+1}|$.  Again, by
    construction, $ |T_{i} \cap S_{i+1}| \ge (1-\eps) |S_{i+1}|$.
    Therefore \(\eps |T_i| \ge (1-\eps) |S_{i+1}| \)
    which implies
    \( |T_i| \ge {\frac{1-\eps}{\eps}} |S_{i+1}| > |S_{i+1}|\)
    since $\eps < 1/2$.  By a similar argument, we can also show that
    $|S_i| > |T_{i}|$.  Consequently,
    $
    |S_0| > |S_{1}| > \ldots > |S_m| = |S_0|
    $ which is a contradiction. 
    So all cycles have length $2$, which 
    implies $\pi_1 = \pi_2^{-1}$.
  \end{proof}
\end{proof}
\section*{Acknowledgments}
We thank Moses Charikar and Ravishankar Krishnaswamy for stimulating
discussions about the problem. We also thank Anup Rao for pointing out
that Lemma~\ref{thm:sandwich} holds with equality.

\bibliographystyle{alpha}
\bibliography{references}

\end{document}